\documentclass[%
superscriptaddress,
amsmath,amssymb,
aps,
prx,
longbibliography,
11pt
]{revtex4-2}

\usepackage{amssymb}
\usepackage{amsthm}
\usepackage{xfrac}
\usepackage{graphicx}
\usepackage{amsmath}
\usepackage{amsfonts}
\usepackage{color}
\usepackage{centernot}
\usepackage{mathtools}
\usepackage{stmaryrd}
\usepackage{xcolor}

\newtheorem{theorem}{Theorem}

\newcommand{\bnabla}{\boldsymbol{\nabla}}

\begin{document}
\title{RG theory of spontaneous stochasticity for Sabra model of turbulence}
\author{Alexei A. Mailybaev} 
\affiliation{Instituto de Matem\'atica Pura e Aplicada -- IMPA, Rio de Janeiro, Brazil}
\email{alexei@impa.br}

\begin{abstract}
We consider fluctuating Sabra models of turbulence, which exhibit the phenomenon of spontaneous stochasticity: their solutions converge to a stochastic process in the ideal limit, when both viscosity and small-scale noise vanish. In this paper, we develop a renormalization group (RG) approach to explain this phenomenon. Here, RG is understood as an exact relation between the stochastic properties of systems with different dissipative and noise terms, in contrast to the Kadanoff-Wilson coarse-graining procedure, which involves small-scale integration. We argue that the stochastic process in the ideal limit is represented as a fixed point of the RG operator. The existence of such a fixed point confirms not only the convergence in the ideal limit, but also the universality of the spontaneously stochastic process, i.e. its independence from the type of dissipation and noise. 
The dominant eigenmode of the linearized RG operator determines the leading correction in the convergence process. 
The RG eigenvalue $\rho \approx 0.84 \exp(2.28i)$ is universal and it turns out to be complex, which explains the rather slow and oscillatory convergence in the ideal limit. 
These universality predictions are accurately confirmed by numerical data.
\end{abstract}

\maketitle

\section{Introduction} 

Since the pioneering work of Lorenz~\cite{lorenz1969predictability} in 1969 it has been argued \cite{leith1972predictability,eyink1996turbulence,boffetta2017chaos,liao2025noise} that tiny errors, even as small as molecular noise~\cite{ruelle1979microscopic}, have a significant impact on the predictability of developed turbulent flows at moderate (turnover) times. Recent numerical studies suggested that although the deterministic description becomes unsuitable, the dynamics is still well defined as a stochastic process~\cite{mailybaev2016spontaneously,thalabard2020butterfly,bandak2024spontaneous,ortiz2025spontaneous}. 
The mathematical formulation of this stochastic process involves a limit in which the noise term in the equations of motion vanishes simultaneously with the viscous term. This limit yields a stochastic process that solves a formally deterministic initial value problem for an ideal flow described by the Euler equations.
Such preservation of randomness in the zero-noise limit is called spontaneous stochasticity.
A similar phenomenon was also detected in the description of particle trajectories carried by deterministic but rough (non-differentiable) velocity fields \cite{bernard1998slow,falkovich2001particles,kupiainen2003nondeterministic,barlet2025spontaneous}. 
Here the spontaneous stochasticity arises in the simultaneous limit of vanishing regularization of the velocity field and noise. 
In this limit trajectories remain stochastic despite they are carried by a formally deterministic velocity field. 
In both cases, a necessary (but not sufficient) condition is the non-uniqueness of solutions to the limiting (ideal) problem, 
and the stochastic description assumes a random choice of such non-unique solutions~\cite{eyink2024space}. Spontaneous stochasticity of trajectories is often called Lagrangian, distinguishing it from Eulerian spontaneous stochasticity of velocity fields.

Eulerian spontaneous stochasticity was conjectured to be universal~\cite{mailybaev2016spontaneously,bandak2024spontaneous,ortiz2025spontaneous}, i.e. the limiting stochastic process does not actually depend on the specific form of dissipative and noise terms. Such a counterintuitive property would be of great practical importance, meaning that small-scale flow details are irrelevant for large-scale probabilistic prediction. Theoretical understanding of this universality in practical fluid models is a very difficult task, but some understanding has already been achieved using toy models. In particular, the existence and universality of spontaneously stochastic solutions has been proven in some specific models~\cite{eyink2020renormalization,mailybaev2023spontaneously,drivas2024statistical}. For a more general class of models represented as dynamics on a fractal space-time lattice, the universality of spontaneous stochasticity has been explained in terms of the renormalization group (RG) approach~\cite{mailybaev2023spontaneous,mailybaev2025rg}.

Our goal here is to extend the RG approach to explain spontaneous stochasticity and its universality in shell turbulence models. Shell models are popular "toy" models that share many non-trivial features of realistic turbulent flows~\cite{frisch1999turbulence,biferale2003shell}, and we focus on the so-called Sabra model~\cite{l1998improved}. Since the RG method was previously developed for discrete-time dynamics on a fractal lattice~\cite{mailybaev2023spontaneous,mailybaev2025rg}, its extension to shell models requires addressing two key issues: the continuous-time formulation and scale invariance with respect to both time and space. This extension is important not only within the framework of shell models, but also as a step toward a possible application of the RG approach to more realistic fluid systems.
Developments in this direction for deterministic models (without noise) were presented in~\cite{mailybaev2024rg}; however, a similar extension to the stochastic setting faces significant mathematical difficulties in defining the RG operator.
Instead, in the present paper we approximate the continuous-time model by dynamics on a properly designed spacetime lattice.
This use of lattice models in situations where analysis in the continuous formulation is problematic is reminiscent of analogous constructions in statistical mechanics and quantum field theory.

The idea is to approximate the Sabra model by a discrete model (akin to an adaptive finite difference scheme) in a way that exactly preserves scaling symmetries and the energy balance. This approximation depends on the parameter $\varepsilon$, so that the solutions of the discrete models converge to the solutions of the Sabra model as $\varepsilon \to 0$. 
The RG approach is formulated in this class of discrete models as an exact relation between their stochastic dynamics for different dissipative and noise terms.
As a consequence, we explain the spontaneously stochastic solutions and their universality  by the existence of a fixed-point attractor of the RG operator. This approach not only explains previous numerical observations, but also yields new predictions. Namely, we show that the leading eigenmode of the linearized RG operator controls the convergence to the spontaneously stochastic solutions. 
This RG eigenvalue is universal, and it also turns out to be complex, which explains the peculiar oscillatory form of convergence. These predictions are verified by detailed numerical simulations.

The paper is organized as follows. In Section \ref{sec2} we formulate the conjecture of spontaneous stochasticity starting with the full fluctuating Navier-Stokes model, and then extending it to the fluctuating Sabra model. Section~\ref{sec1C} introduced the discrete-time fluctuating shell models. For such models, the RG operator is introduced in Section~\ref{sec4}, and the spontaneous stochasticity is explained as the fixed-point RG attractor in Section~\ref{secRGFP}. The RG approach is limited to a class of regularizations and noise, which preserve symmetries of the ideal system; we call such regularizations canonical. The physical (viscous) regularization does not belong to this class. We discuss the extension of the RG theory to the viscous and thermal noise conditions of the original Sabra model in Section~\ref{sec5}. The paper ends with the Conclusion. Appendix contains some extra information and technical details.

\section{The concept of (Eulerian) spontaneous stochasticity} \label{sec2}

\subsection{Fluctuating Navier-Stokes equations} 
\label{subsecSpSt_A}

As a motivation for this work, we discuss in this section the incompressible three-dimensional Navier-Stokes equations with a microscopic (Landau-Lifshitz) fluctuation term. In the dimensionless form, these equations take the form~\cite{landau1959fluid,bandak2024spontaneous}
	\begin{equation}
	\label{eq1_1}
	\partial_t \mathbf{u} + \mathbf{u} \cdot \bnabla \mathbf{u} = -\bnabla p + \mathrm{Re}^{-1} \Delta \mathbf{u} 
	+ \sqrt{\Theta}\, \bnabla \cdot \boldsymbol{\xi} + \mathbf{f}, \quad \bnabla \cdot \mathbf{u} = 0.
	\end{equation}
Here $\mathbf{u}(\mathbf{x},t)$ is a velocity field, which we consider in a triply $2\pi$-periodic domain $\mathbf{x} \in \mathbb{T}^3$, $p(\mathbf{x},t)$ is a pressure, and $\mathbf{f}(\mathbf{x},t)$ is a deterministic forcing applied at large scales. The fluctuating stress $\boldsymbol{\xi}$ is modeled as a Gaussian random matrix field with mean zero and covariance $\langle \xi_{ij}(\mathbf{x},t) \xi_{kl}(\mathbf{x}',t') \rangle = \left( \delta_{ik}\delta_{jl}+\delta_{il}\delta_{jk}-\frac{2}{3}\delta_{ij}\delta_{kl} \right) \delta^3(\mathbf{x}-\mathbf{x}')\delta(t-t')$. System (\ref{eq1_1}) depends on two dimensionless parameters: $\mathrm{Re}$ and $\Theta$. Here $\mathrm{Re} = UL/\nu$ is the Reynolds number composed of the kinematic viscosity $\nu$ dividing the product of characteristic (large-scale) velocity $U$ and scale $L$. The second parameter is expressed in terms of thermodynamic properties of the fluid as $\Theta = 2\nu k_BT/\rho L^4U^3$, where $k_B$ is Boltzmann’s constant, $T$ is absolute temperature, and $\rho$ is mass density. 
In addition, one can consider the Galerkin cutoff to Fourier modes with wave numbers $|\mathbf{k}| \le 2\pi/\delta$, i.e. equate all modes with wave numbers $|\mathbf{k}| > 2\pi/\delta$ to zero.
Here $\delta = \ell_{\mathrm{mic}}/L$ refers to the microscopic (molecular) scale $\ell_{\mathrm{mic}}$ at which the macroscopic description (\ref{eq1_1}) cease to be valid. Mathematically, the Galerkin truncation conveniently reduces Eq.~(\ref{eq1_1}) to a very large but still finite system of stochastic differential equations considered in the Fourier representation, thereby ensuring the well-posedness~\cite{bandak2022dissipation,eyink2024space}. In particular, the gradient of the noise term in Eq.~(\ref{eq1_1}) is well defined in this truncated Fourier representation.

Our focus will be on solving the initial value problem, where the initial condition
	\begin{equation}
	\label{eq1_1IC}
 	\mathbf{u}(\mathbf{x},0) = \mathbf{u}_0(\mathbf{x})
 	\end{equation}
is a given (deterministic) velocity field. 
Due to the noise term present in system (\ref{eq1_1}), the solution is a stochastic process, which determines the probability distribution of velocity fields at times $t > 0$. 
Note that stochasticity is introduced only through microscopic noise, while both the large-scale force $\mathbf{f}$ and the initial state $\mathbf{u}_0$ are deterministic.

In typical developed turbulent flows, all three dimensionless parameters, $\mathrm{Re}^{-1}$, $\Theta$ and $\delta$, are extremely small.
For example, in the atmospheric boundary layer, their estimates (by order of magnitude) are \cite{garratt1994atmospheric,bandak2024spontaneous}
	\begin{equation}
	\label{eq1_3}
	\mathrm{Re}^{-1} \sim 10^{-8}, \quad \Theta \sim 10^{-38}, \quad \delta \sim 10^{-10}.
	\end{equation}
The extreme smallness of these parameters suggests that the solution can be approximated by considering their vanishing limit.  
The traditional approach is to take the limits $\Theta \to 0$ and $\delta \to \Theta$ first, which yield the deterministic Navier--Stokes equations.
Then one can take the inviscid limit $\mathrm{Re}^{-1} \to 0$ (usually written as $\mathrm{Re} \to \infty$), which leads to the Euler equations for an ideal fluid as
	\begin{equation}
	\label{eq1_1Euler}
	\partial_t \mathbf{u} + \mathbf{u} \cdot \bnabla \mathbf{u} = -\bnabla p+\mathbf{f}, \quad \bnabla \cdot \mathbf{u} = 0.
	\end{equation}
The problem is that, based on numerical evidence, turbulent solutions of the Navier–Stokes equations do not appear to converge as $\mathrm{Re} \to \infty$.

The spontaneous stochasticity (SpSt) hypothesis suggests that a correct asymptotic description of developed turbulence requires a different formulation of the limit.
Namely, it is necessary to consider the limit in which all three small parameters, $\mathrm{Re}^{-1}$, $\Theta$ and $\delta$, tend to zero simultaneously~\cite{bandak2024spontaneous,ortiz2025spontaneous}.
For example, the limit compatible with the values (\ref{eq1_3}) can be formulated as 
	\begin{equation}
	\label{eq1_4}
	\mathrm{Re}^{-1} \to 0, \quad \Theta = \mathrm{Re}^{-\beta} \to 0, \quad \delta = \mathrm{Re}^{-\gamma} \to 0,
	\end{equation}
with the exponents $\beta = 38/8$ and $\gamma = 10/8$. 
This limit leads directly to the Euler equations~(\ref{eq1_1Euler}), and we refer to any limit of this type as an ideal limit, generalizing the concept of the inviscid limit to include vanishing noise.
The formulation (\ref{eq1_4}) does not mean that the physical parameters vary in a fixed proportion. Rather, it provides a proper mathematical approximation to the (very small, but finite) physical parameters (\ref{eq1_3}).

We formulate the properties expected from the spontaneously stochastic limit as~\cite{mailybaev2015stochastic,mailybaev2016spontaneously,bandak2024spontaneous,ortiz2025spontaneous}
	\begin{itemize}
	\item[(\textit{i})] \textit{Convergence}: the solution of Eqs.~(\ref{eq1_1}) and (\ref{eq1_1IC}) converges in distribution in the ideal limit.
	\item[(\textit{ii})] \textit{Stochasticity}: the limiting dynamics is not deterministic, but is a stochastic process.
	\item[(\textit{iii})] \textit{Spontaneity}: realizations of the limiting stochastic process are solutions 
	of the deterministic Euler system (\ref{eq1_1Euler}) and (\ref{eq1_1IC}). 
	\item[(\textit{iv})] \textit{Universality}: this process does not depend on the formulation of the ideal limit.
	\end{itemize}
These conditions imply that the stochastic process solving the fluctuating Navier-Stokes system converges, in the ideal limit, to a stochastic process that solves the formally deterministic initial-value problem for the Euler equations.
A necessary requirement for this limiting process to exist is that the Euler system (\ref{eq1_1Euler}) and (\ref{eq1_1IC}) is ill-posed, namely, has non-unique (weak) solutions~\cite{de2010admissibility,daneri2021non}. 
However, the mere existence of non-unique solutions does not, by itself, guarantee either convergence or a stochastic form of the limit.
The third property is natural and can be justified mathematically~\cite{eyink2024space}. 
Universality is a much stronger property: it means that the limiting stochastic process is identical for a wide class of regularizations, noise types, and formulations of the ideal limit.
For example, the limit should be independent of the exponents $\beta$ and $\gamma$ in Eq.~(\ref{eq1_4}), at least within certain intervals of their values. The reader should regard our characterization of spontaneous stochasticity more as a guideline than as a proposal for a precise mathematical formulation. Indeed, its interpretation and validity for the three-dimensional fluctuating Navier-Stokes system remain poorly understood, while simplified models suggest a broader range of possible scenarios~\cite{mailybaev2025rg}.

Detailed and convincing evidence for spontaneous stochasticity is provided by shell models of turbulence, which are toy models of the Navier-Stokes system, as we discuss next. 
As for the fluctuating Navier-Stokes system itself, convergence to a statistical solution in the ideal limit has been observed numerically for two-dimensional flows starting from the initial state of a vortex sheet, in close connection with the Kelvin-Helmholtz instability~\cite{fjordholm2016computation,thalabard2020butterfly}. 
In particular, in the work \cite{thalabard2020butterfly} the universality of the spontaneously stochastic solution was tested for two types of regularization: viscous and point vortex approximation. A recent study~\cite{ortiz2025spontaneous} confirms spontaneous stochasticity for the fluctuating Navier-Stokes system on a logarithmic lattice, which is another type of toy model.

\subsection{Fluctuating Sabra shell model}

Shell models of turbulence are toy models designed to mimic the multiscale dynamics of the Navier-Stokes system~\cite{frisch1999turbulence,biferale2003shell}.
In this paper, we consider the fluctuating Sabra model~\cite{l1998improved,bandak2022dissipation}, which features many properties of turbulence in the fluctuating Navier-Stokes equations (\ref{eq1_1}). In particular, it confirms the spontaneous stochasticity~\cite{bandak2024spontaneous}.

Shell models mimic the physical space by a discrete sequence of scales $\ell_n = 2^{-n}$ for integer shell numbers $n$. The respective wavenumbers are defined as $k_n = 2^n$. Each scale (shell $n$) is described by a complex-valued variable $u_n(t)$, which is called shell velocity and depends on time $t$. The fluctuating dynamics of the Sabra model is governed by the equations
	\begin{equation}
	\label{eq1_5}
	\frac{du_n}{dt} = B_n-\mathrm{Re}^{-1} k_n^2 u_n + \sqrt{\Theta}\, k_n \eta_n, \quad n = 1,\ldots,N,
	\end{equation}
where $\mathrm{Re}$ and $\Theta$ have the same meaning as for the Navier-Stokes system (\ref{eq1_1}).
The nonlinear terms $B_n$ mimic the quadratic nonlinearities of the Navier-Stokes system and are defined as
	\begin{equation}
	\label{eq1_6}
	B_n
	= i \Big( k_{n+1}u_{n+2}u_{n+1}^*-\frac{1}{2}k_{n}u_{n+1}u_{n-1}^*+\frac{1}{2}k_{n-1}u_{n-1}u_{n-2} \Big),
	\end{equation}
where $i$ is the imaginary unit and the asterisks denote complex conjugation. 
These terms are designed to conserve the inviscid invariants of total energy $\mathcal{E} = \sum |u_n|^2$ and helicity $\mathcal{H} = \sum (-1)^nk_n|u_n|^n$ analogous to the inviscid conservation laws of the ideal Euler equations~\cite{l1998improved}. The fluctuations are represented by complex white-noise terms $\eta_n(t)$ with the covariance $\langle \eta_n(t)\eta^*_{n'}(t') \rangle = 2\delta_{n n'} \delta(t-t')$. Additionally, we introduce the cutoff shell $N$ and define the respective cutoff scale $\delta = \ell_N = 2^{-N}$. Hence, the velocities at larger shell numbers are all set to zero:
	\begin{equation}
	\label{eq1_7cut}
 	u_n(t) \equiv 0 \quad \textrm{for} \ \ n > N. 
  	\end{equation}
The deterministic (large-scale) forcing is introduced by setting 
	\begin{equation}
	\label{eq1_7}
	u_0(t) \equiv b_0, \quad u_{-1}(t) \equiv b_{-1},
 	\end{equation}
where $b_0$ and $b_{-1}$ are given complex numbers. 
For simplicity, in this paper we consider constant-in-time boundary conditions (\ref{eq1_7}).
Note that the boundary variables enter into the nonlinear terms (\ref{eq1_6}) for $n = 1$ and $2$. The initial conditions are defined as
	\begin{equation}
	\label{eq1_8}
	u_n(0) = a_n, \quad n = 1,\ldots,N,
	\end{equation}
for given initial values $a_1,\ldots,a_N$. 

The solution $u_1(t),\ldots,u_N(t)$ of the initial value problem is a stochastic process that solves the system of stochastic differential equations (\ref{eq1_5}) for given initial conditions (\ref{eq1_8}) and fixed parameters $\mathrm{Re}$, $\Theta$, $N$ and boundary conditions. The concept of spontaneous stochasticity is expressed by the same four conditions (\textit{i--iv}) from Section~\ref{subsecSpSt_A}, which are formulated for the ideal limit. Equations (\ref{eq1_4}) provide a particular form of such ideal limit, in which the cutoff shell can be defined as the integer part $N = [\log_2 1/\delta]$. The numerical verifications of spontaneous stochasticity in shell models were reported in \cite{mailybaev2015stochastic,mailybaev2016spontaneous,mailybaev2016spontaneously,mailybaev2017toward,biferale2018rayleigh} considering the stochastic formulation with random fluctuations of viscosity or initial conditions, and in \cite{bandak2022dissipation} considering the small-scale white-noise noise as in Eq.~(\ref{eq1_5}). 
These works confirmed that evolution from a given deterministic initial state converges and remains a stochastic process in the ideal limit, and that this limit process is independent of the details of the regularization and the type of noise. We perform a similar numerical study below in Section~\ref{subsec_NSpSt},

In the remainder of this subsection, we provide some additional information about the ideal system that will be useful for our further analysis. The ideal model equations are obtained from Eqs.~(\ref{eq1_5}) and (\ref{eq1_6}) by setting all small parameters, $\mathrm{Re}^{-1}$, $\Theta$ and $\delta$, to zero. This yields
	\begin{equation}
	\label{eq1_5ID}
	\frac{du_n}{dt} = i \Big( k_{n+1}u_{n+2}u_{n+1}^*-\frac{1}{2}k_{n}u_{n+1}u_{n-1}^*+\frac{1}{2}k_{n-1}u_{n-1}u_{n-2} \Big).
	\end{equation}
The ideal system possesses the time-scaling symmetry 
	\begin{equation}
	\label{eq1_ST}
	u_n(t) \ \mapsto \ \alpha u_n(\alpha t), \quad \alpha > 0,
	\end{equation}
meaning that Eq.~(\ref{eq1_5ID}) is invariant under the transformation $\tilde{u}_n(t) = \alpha u_n(\alpha t)$ for all $n$. Another (space-scaling) symmetry is expressed as
	\begin{equation}
	\label{eq1_SS}
	u_n(t) \ \mapsto \ 2 u_{n+1}(t),
	\end{equation}
which means that Eq.~(\ref{eq1_5ID}) is invariant under the transformation $\tilde{u}_n(t) = 2 u_{n+1}(t)$. Note that both symmetries (\ref{eq1_ST}) and (\ref{eq1_SS}) are broken in the full model (\ref{eq1_5}).

It is convenient to split the nonlinear term (\ref{eq1_6}) into two parts $B_n = B_n^-+B_n^+$ as
	\begin{equation}
	\label{eq1B_2}
	B_n^- =
	i \big(-k_{n}u_{n+1}u_{n-1}^*+\frac{1}{2}k_{n-1}u_{n-1}u_{n-2} \big), 
	\quad
	B_n^+ =
	i \big( k_{n+1}u_{n+2}u_{n+1}^*+\frac{1}{2}k_{n}u_{n+1}u_{n-1}^* \big).
	\end{equation}
For the ideal system (\ref{eq1_5ID}), these terms determine the energy balance at shell $n$ in the form
	\begin{equation}
	\label{eq1B_3}
	\frac{d}{dt} |u_n|^2 = \Pi_{n-1}-\Pi_{n},
	\end{equation}
where we introduced the energy flux from shell $n$ to $n+1$ as
	\begin{equation}
	\label{eq1B_3_Pi}
	\Pi_{n} = -2\, \mathrm{Re} \left( u_n^*B_n^+ \right) 
	\equiv 2\, \mathrm{Re} \left( u_{n+1}^*B_{n+1}^- \right).
	\end{equation}
Note the last identity, which follows from the expressions in Eq.~(\ref{eq1B_2}).

\subsection{Fluctuating LES shell model}

Let us introduce another version of the fluctuating shell model, in which viscosity and noise are defined in the spirit of turbulence closures of the Large Eddy Simulation (LES)~\cite{pope2000turbulent}. We call it the LES shell model. 
This model will be important for our explanation of spontaneous stochasticity using the RG approach.

In our LES shell model, we consider the ideal equations for the shells $n = 1,\ldots,N-2$, while dissipative and noise terms are added to the shells $N-1$ and $N$. 
The viscous terms are taken as $-\nu_n k_n^2 u_n$ with the effective (eddy) viscosity $\nu_n = D |u_n|/k_n$, where $D > 0$ is a fixed parameter. 
Similarly, we introduce  effective (eddy) noise terms on dimensional basis as $\sqrt{\Theta_n}\, k_n \eta_n$, where $\Theta_n = C^2 |u_n|^3/k_n$ with the fixed parameter $C > 0$. 
In summary, the governing equations of the LES shell model take the form
	\begin{equation}
	\label{eq1_5SGS}
	\frac{du_n}{dt} = \left\{ \begin{array}{ll}
	B_n, & n = 1,2,\ldots,N-2; \\
	B_n-D k_n |u_n| u_n+ C k_n^{1/2} |u_n|^{3/2} \eta_n,& n = N-1,N.
	\end{array}\right.
	\end{equation}
The boundary and initial conditions are considered in the same form; see Eqs.~(\ref{eq1_7cut})--(\ref{eq1_8}). 

The ideal limit is defined as $N \to \infty$, in which the model (\ref{eq1_5SGS}) reduces to the same ideal system (\ref{eq1_5ID}). 
Like the fluctuating Sabra model, the fluctuating LES shell model (\ref{eq1_5SGS}) is spontaneously stochastic: it defines a stochastic process in the ideal limit. These limits appear to be the same in both models, as confirmed by numerical simulations in the next Section~\ref{subsec_NSpSt}.
	
The important property of the LES shell model is that Eqs.~(\ref{eq1_5SGS}) are invariant with respect to the time-scaling (\ref{eq1_ST}). Moreover, the space-scaling symmetry is also preserved when written as
	\begin{equation}
	\label{eq1_SSN}
	u_n(t),\, N \ \mapsto \ 2 u_{n+1}(t),\, N-1,
	\end{equation}
meaning that Eqs.~(\ref{eq1_5SGS}) are invariant under the transformations $\tilde{u}_n(t) = 2 u_{n+1}(t)$ and $\tilde{N} = N-1$. 

\subsection{Numerical observation of spontaneous stochasticity}
\label{subsec_NSpSt}

For numerical simulations we consider the fluctuating Sabra model (\ref{eq1_5}) with different Reynolds numbers and the other parameters given by Eq.~(\ref{eq1_4}) with the exponents $\beta = 2.2$ and $\gamma = 1.2$. The integration is carried out using the split-step technique: the first step integrates the nonlinear term $B_n$ using the fourth-order Runge-Kutta method and the second step integrates the viscous and noise terms as the Ornstein-Uhlenbeck process. The time step is selected as $\Delta t = 0.01\, \mathrm{Re}^{-1/2}$, which is $1\%$ of the turnover time at the Kolmogorov viscous scale as approximated by the K41 theory~\cite{frisch1999turbulence}. We take the boundary conditions $b_{-1} = b_0 = 1$ and the initial conditions $\mathbf{a} = (1,0,0,\ldots)$. The statistical analysis is performed by averaging over $4 \times 10^4$ stochastic simulations in each setting. 

Figure~\ref{fig1}(a) presents the result of a single simulation for $\mathrm{Re} = 10^9$. The solution demonstrates a ``laminar'' evolution till time $t_b \approx 0.8$, after which the dynamics becomes turbulent. The transition time $t_b$ corresponds to the blowup of the ideal system~\cite{dombre1998intermittency,mailybaev2012renormalization}. Figure~\ref{fig1}(b) shows the absolute values $|u_n|$ depending on the shell number at the final time $T = 1.5$. One can recognize the forcing, inertial, dissipative and thermal ranges of scales~\cite{frisch1999turbulence,bandak2022dissipation}. 

\begin{figure}[tp]
\centering
\includegraphics[width=0.83\textwidth]{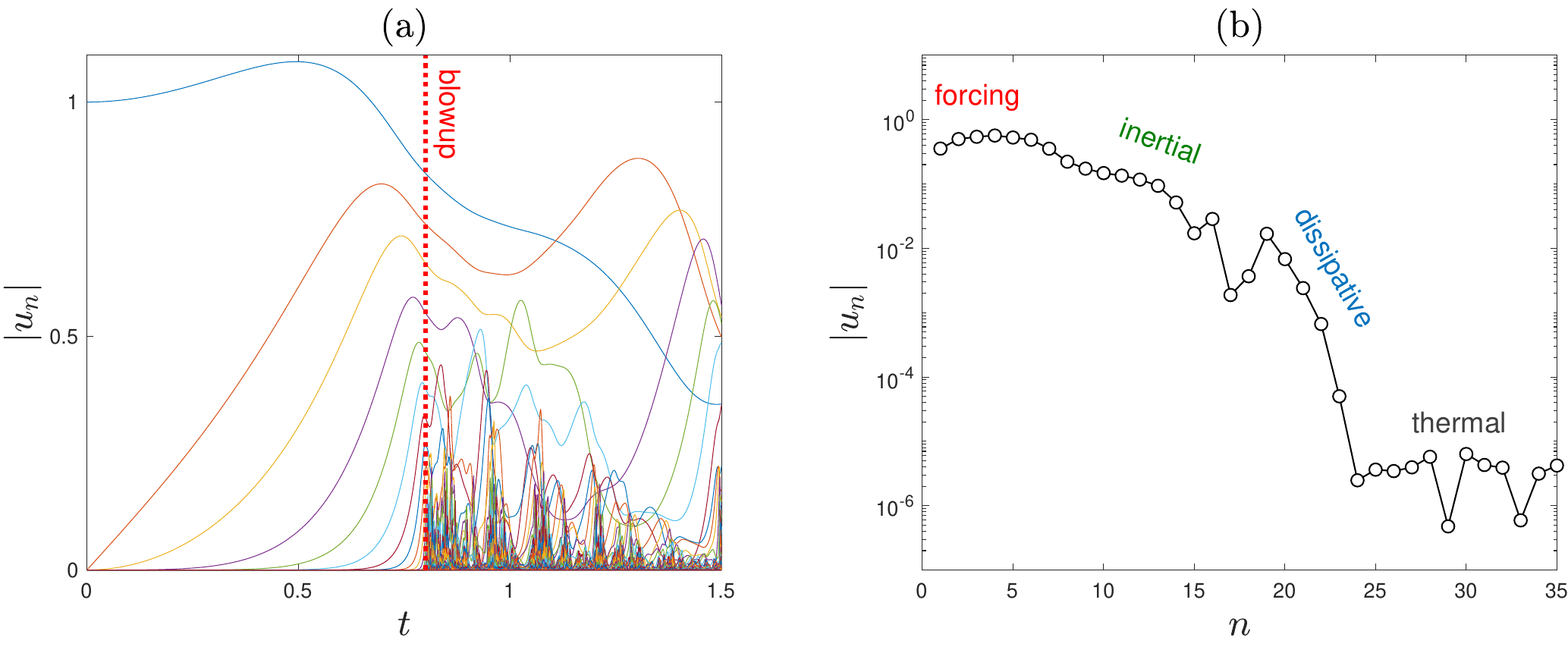}
\caption{(a) Evolution of the fluctuating Sabra model for $\mathrm{Re} = 10^9$. Turbulent dynamics begins at the blowup time $t_b \approx 0.8$ (vertical line). Each curve corresponds to a different shell $n = 1,\ldots,N$ with $N = 35$. (b) Turbulent state at the final time $T = 1.5$ for the realization from the left panel.
}
\label{fig1}
\end{figure}

We now verify that the solution is spontaneously stochastic. Figure \ref{fig2}(a) shows the absolute velocity $|u_1(t)|$ in fifty different simulations with the same parameters as in Fig.~\ref{fig1}(a). These graphs evidence that sample solutions are very close before and diverge after the blowup time. Figure~\ref{fig2}(b) computes the time dependence of respective standard deviations for different Reynolds numbers. 
As the Reynolds number increases from  $\mathrm{Re} = 10^5$ to $10^9$, these standard deviations tend to zero at times $t \le t_b$, but converge to positive values at $t > t_b$. 
Hence, the limiting solution remain stochastic in the ideal limit after the blowup time. Further confirmation is given in Figs.~\ref{fig2}(c,d), where we plot probability density functions (PDFs) of $\mathrm{Re}\, u_1(T)$ and $\mathrm{Re}\, u_2(T)$ at time $T = 1.5$. These graphs demonstrate the convergence to non-Dirac distributions. 

\begin{figure}[tp]
\centering
\includegraphics[width=0.8\textwidth]{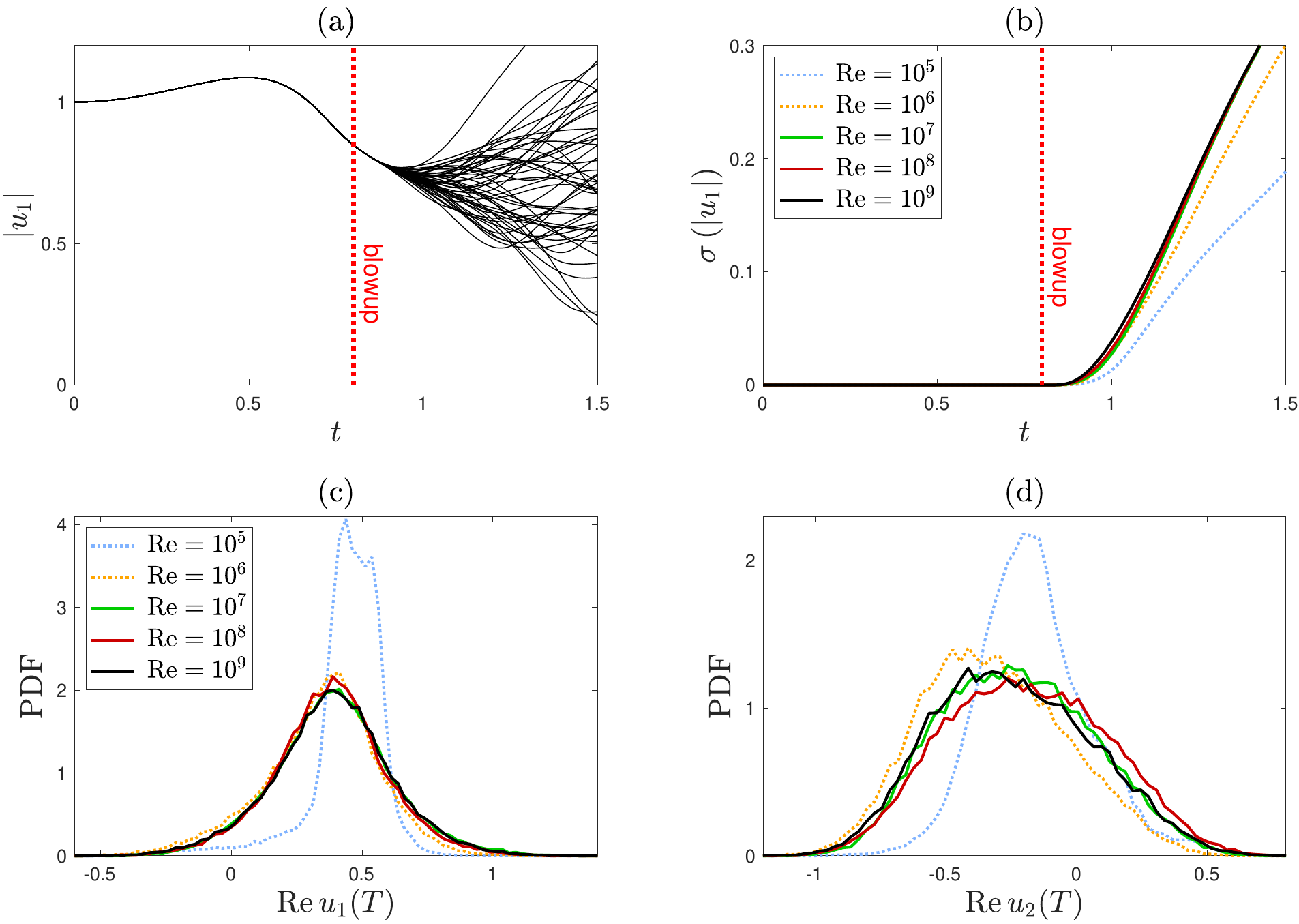}
\caption{(a) Absolute shell velocities $|u_1(t)|$ in fifty realizations of the fluctuating Sabra model for $\mathrm{Re} = 10^9$ with the same initial and boundary conditions. (b) Standard deviations of $|u_1(t)|$ for the Reynolds numbers $\mathrm{Re} = 10^5,\ldots,10^9$. (c,d) PDFs of $\mathrm{Re}\,u_1(T)$ and $\mathrm{Re}\,u_2(T)$ at $T = 1.5$ for the same sequence of Reynolds numbers. The PDFs are calculated using the histogram approach for $4 \times 10^4$ simulations.}
\label{fig2}
\end{figure}

\begin{figure}[tp]
\centering
\includegraphics[width=0.72\textwidth]{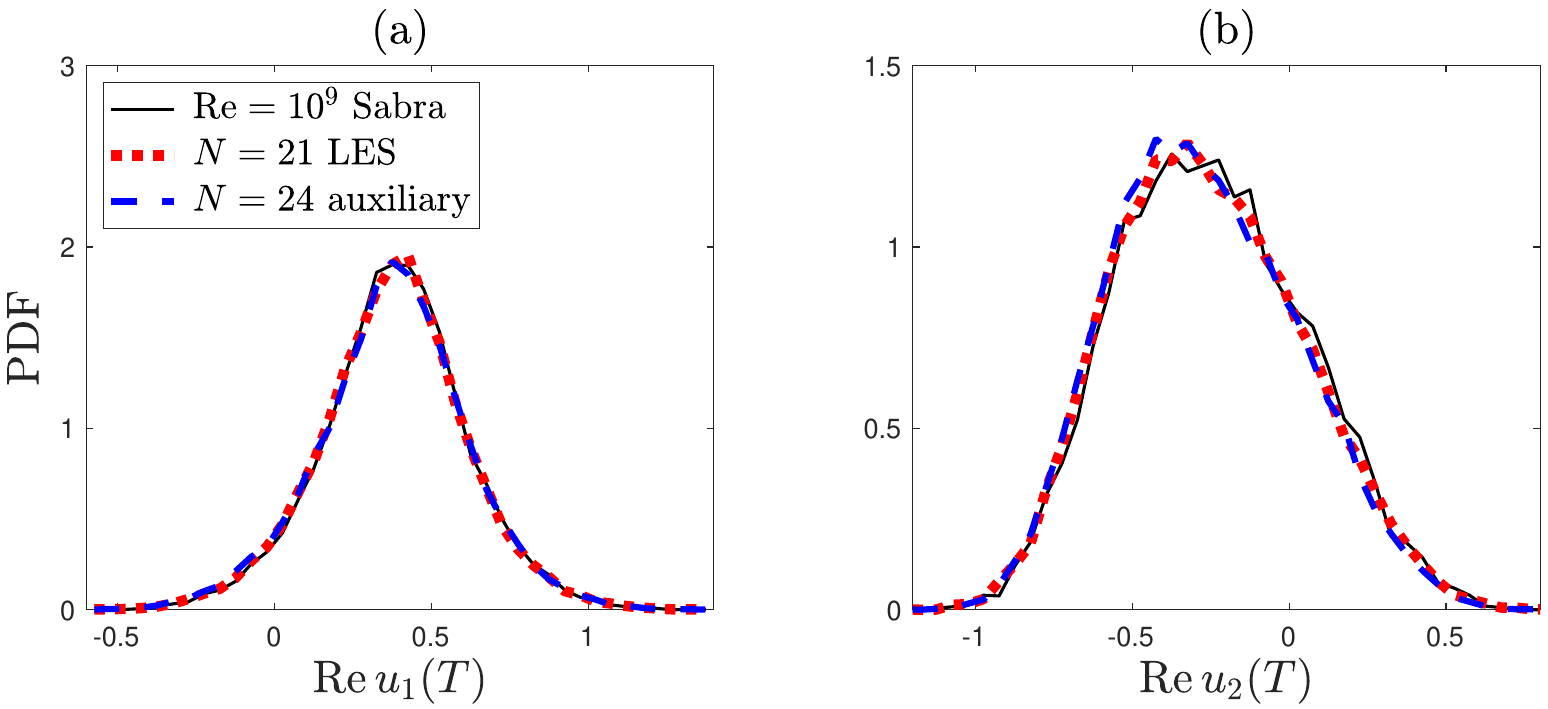}
\caption{Universality of spontaneous stochasticity for three different regularizations: the fluctuating Sabra model (\ref{eq1_5}), the fluctuating LES model (\ref{eq1_5SGS}) and the auxiliary model (\ref{eqRGV_1}). Comparison of PDFs for (a) $\mathrm{Re}\,u_1(T)$ and (b) $\mathrm{Re}\,u_2(T)$.}
\label{fig4}
\end{figure}

In order to verify the universality of spontaneous stochasticity, we performed a similar analysis for the fluctuating LES shell model (\ref{eq1_5SGS}) with $D = 1$, $C = 0.5$. This model is characterized by different dissipation and noise mechanisms, as well as a different definition of the ideal limit, $N \to \infty$. 
We performed numerical simulations with the same initial and boundary conditions. Our results confirmed that statistical properties at large scales converge, and that the limiting statistics are the same as for the fluctuating Sabra model. This is demonstrated in Fig.~\ref{fig4} showing the results for the fluctuating LES shell model with $N = 21$ and the fluctuating Sabra model with $\mathrm{Re} = 10^9$.

We note that although convergence to spontaneously stochastic solutions is clearly observed, it occurs relatively slowly: the accurate convergence requires very large (though still physically meaningful) values of Reynolds numbers. 
Even for shell models, accurate analysis of spontaneous stochasticity requires significant computational resources.

\section{Discrete-time fluctuating shell models}
\label{sec1C}

The Sabra shell model, although much simpler than the Navier-Stokes system, is still difficult to analyze numerically and theoretically. 
In this model, space is discretized, but time remains continuous and requires finite-difference approximation by numerical schemes. We now introduce a class of discrete-time models that, on the one hand, approximate the Sabra model to any accuracy (in the spirit of the finite-difference numerical method). 
On the other hand, such discrete-time models can themselves be considered as proper toy models of turbulence, since they preserve precisely the key properties: space-time scale invariance and energy conservation.
These models have a considerable numerical advantage since their simulations are much faster and very accurate. 
Most importantly, their conservative and symmetry-preserving structure, which is not retained by standard numerical schemes, allows for an explicit and transparent formulation of the renormalization-group theory.

\subsection{Discrete-time algorithm}
\label{subsec_DTSM}

Our discrete-time Sabra (DT-Sabra) model depends on the fixed discretization parameter $\varepsilon > 0$ and represents an adaptive first-order numerical scheme for the Sabra equations. This model exactly preserves both the scale invariance and the energy balance of the ideal system. The corresponding space-time lattice is shown schematically in Fig.~\ref{fig5}. Time steps at each shell depend on the state of the system and obey the hierarchical rule: for any node $(n,t)$ there exist simultaneous nodes $(m,t)$ at shells $m > n$. 

\begin{figure}[tp]
\centering
\includegraphics[width=0.85\textwidth]{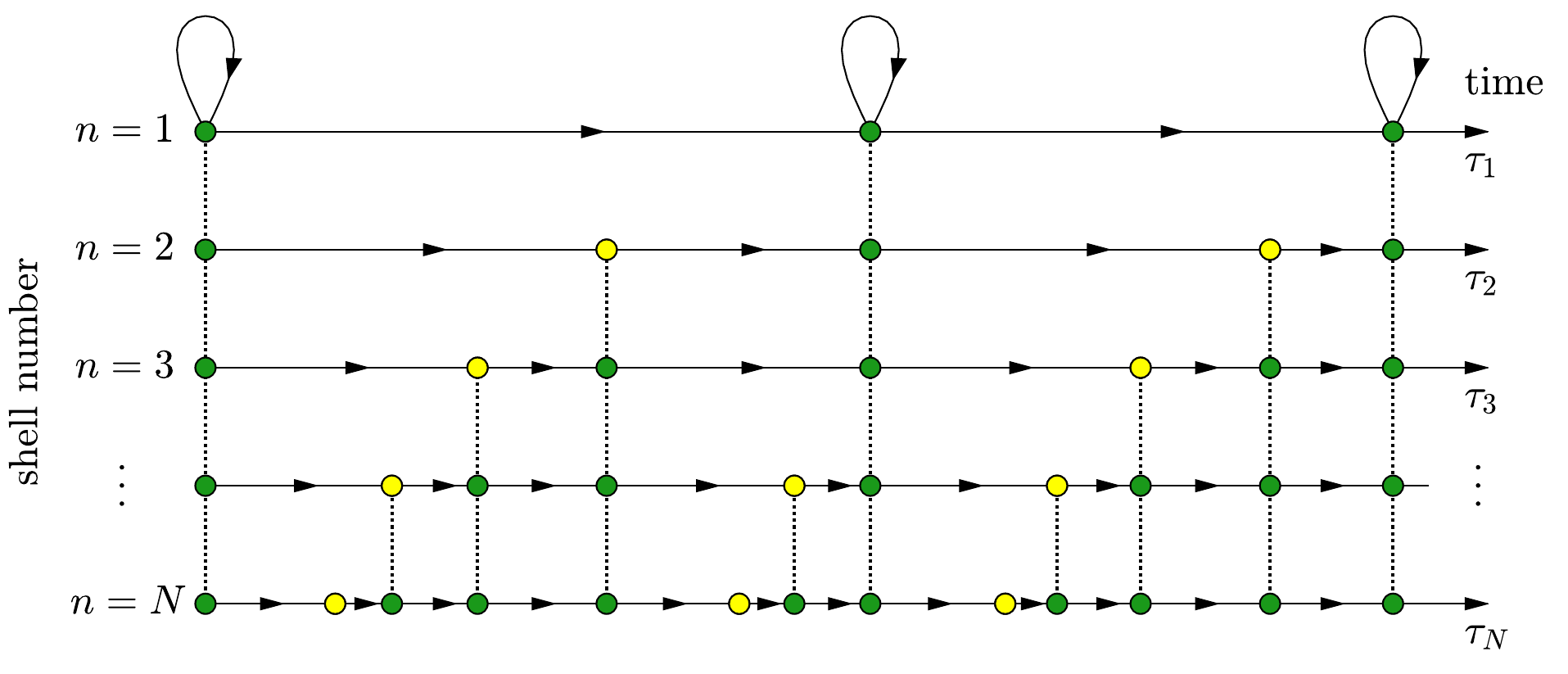}
\caption{Space-time lattice of the discrete-time shell model. Each arrow corresponds to one step of the integration algorithm, in which the time intervals are chosen adaptively at each shell $n$. This algorithm provides a first-order numerical scheme for the original Sabra model at the green nodes.}
\label{fig5}
\end{figure}

We now describe the solution in a given time interval $0 \le t \le T$ as an algorithm, in which every step advances time at a single node as shown by arrows in Fig.~\ref{fig5}. 
We define an intermediate state of the system as $(\mathbf{u},\boldsymbol{\tau})$, where $\mathbf{u} = (u_1,\ldots,u_N)$ is a vector of the shell variables and $\boldsymbol{\tau} = (\tau_1,\ldots,\tau_N)$ is a vector of respective times: each $u_n$ corresponds to the node $(n,\tau_n)$. 
Recall that $u_{-1}$ and $u_0$ are given by the boundary conditions (\ref{eq1_7}) and $u_n = 0$ for $n > N$.
According to the initial conditions (\ref{eq1_8}), the algorithm is initiated at the state $(\mathbf{u}_{\mathrm{ini}},\mathbf{0})$ with $\mathbf{u}_{\mathrm{ini}} = (a_1,\ldots,a_N)$ and $\mathbf{0} = (0,\ldots,0)$, and stops at the final state $(\mathbf{u}_{\mathrm{fin}},\boldsymbol{\tau}_{\mathrm{fin}})$ with $\boldsymbol{\tau}_{\mathrm{fin}} = (T,\ldots,T)$.

Given an intermediate state $(\mathbf{u},\boldsymbol{\tau})$, one step of the algorithm provides a new state $(\mathbf{u}_{\mathrm{new}},\boldsymbol{\tau}_{\mathrm{new}})$ as described below. 
Let $n$ as the smallest shell number such that $\tau_n$ is the minimal component of $\boldsymbol{\tau} = (\tau_1,\ldots,\tau_N)$. 
We define the corresponding time step as
	\begin{equation}
	\Delta \tau_n = \left\{ \begin{array}{ll} 
	\min \left( \varepsilon T_1,\, T-\tau_1 \right), & n = 1, \\[2pt]
	\min \left( \varepsilon T_n,\, \tau_{n-1}-\tau_n \right), & n \ge 2, 
	\end{array}\right. 
	\label{eqDTS_TSt}
	\end{equation}
where $\varepsilon T_n$ represents the $\varepsilon$-fraction of the local turnover time defined as
	\begin{equation}
	T_n = \frac{1}{k_n}
	\bigg(\sum_{m \ge n-2} |u_m|^2\bigg)^{-1/2}.
	\label{eqDTS_0T}
	\end{equation}
The second term in Eq.~(\ref{eqDTS_TSt}), either $T-\tau_1$ or $\tau_{n-1}-\tau_n$, is necessary to maintain the hierarchical structure of the lattice and the stopping condition at the final time $T$.
Then components of new time vector $\boldsymbol{\tau}_{\mathrm{new}}$ are defined as
	\begin{equation}
	\label{eqDTS_Tnew}
	\tau_{\mathrm{new},m} = \left\{ \begin{array}{ll}
		\tau_n+\Delta \tau_n, & m = n; \\
		\tau_m, & m \ne n;
	\end{array}\right.
	\end{equation}
where the time is advanced only at the shell $n$; see Fig.~\ref{fig5}.

The new state $\mathbf{u}_{\mathrm{new}}$ is determined by a sequence of two (split-step) transformations with an intermediate state $\mathbf{u}'$. The latter is determined by the nonlinear interactions as
	\begin{equation}
	\label{eqDTL_b}
	u'_{n} = u_{n}+B_n^+\left[ \frac{\mathbf{u}+\mathbf{u}'}{2} \right] \Delta \tau_n, \quad 
	u'_{n+1} = u_{n+1}+B_{n+1}^-\left[ \frac{\mathbf{u}+\mathbf{u}'}{2} \right] \Delta \tau_n.
	\end{equation}
Note that this relation modifies only the two components at shells $n$ and $n+1$, while all other components are left unchanged, so that $u'_m = u_m$ for $m \ne n,n+1$.
Here the nonlinear terms (\ref{eq1B_2}) are evaluated for the mean state $(\mathbf{u}+\mathbf{u}')/2$. The second relation in Eq.~(\ref{eqDTL_b}) is omitted at the cutoff shell $n = N$. 
Relations in Eq.~(\ref{eqDTL_b}) define the implicit finite-difference scheme for the nonlinear interaction in the Sabra model (\ref{eq1_5}), in which the two parts $B_n^-$ and $B_n^+$ of the nonlinear term are treated independently. This scheme is designed in such a way that the energy is conserved exactly; see Appendix \ref{subsec_A1}, where Eqs.~(\ref{eqDTL_b}) are solved analytically. The dissipation and noise are introduced by solving the Ornstein-Uhlenbeck process
	\begin{equation}
	\label{eq1_9}
	\frac{du_n}{dt} = -\nu_n u_n + \sigma_n \eta_n, \quad
	\nu_n = \mathrm{Re}^{-1} k_n^2, \quad \sigma_n = \sqrt{\Theta}\, k_n,
	\end{equation}
with the initial condition $u_n(\tau_n) = u'_{n}$ and the time interval $\Delta \tau_n$. Its solution has the form~\cite{oksendal2013stochastic}
	\begin{equation}
	u_{\mathrm{new},n} = u'_{n} e^{-\nu_n\Delta \tau_n}+
	\sigma_n \left(\frac{1-e^{-2\nu_n\Delta \tau_n}}{\nu_n}\right)^{1/2} z,
	\label{eq1_10}
	\end{equation}
where $z$ is a standard complex Gaussian random variable, which is independent at every step of the algorithm. 
Expression (\ref{eq1_10}) defines the $n$-th component of the new state $\mathbf{u}_{\mathrm{new}}$, while the other components are set to $u_{\mathrm{new},m} = u'_{m}$ for $m \ne n$. 

Additionally, whenever we encounter a synchronized time vector, $\boldsymbol{\tau} = (t,\ldots,t)$ with $0 \le t < T$, we perform one extra nonlinear transformation $(\mathbf{u},\boldsymbol{\tau}) \mapsto (\mathbf{u}_{\mathrm{new}},\boldsymbol{\tau})$ as
	\begin{equation}
	u_{\mathrm{new},n} = u_n+\left\{ \begin{array}{ll}
	B_1^-[ \mathbf{u}] \Delta \tau_1, & n = 1; \\
	0,& n \ge 2.
	\end{array}\right.
	\label{eq1_10BC}
	\end{equation}
This transformation precedes the usual step described earlier and accounts for term $B_1^-$, which otherwise would not be captured by the relations (\ref{eqDTL_b}); see the loop arrows in Fig.~\ref{fig5}.

We refer the reader to Appendix~\ref{app_alg} for an explicit end-to-end description of the algorithm, ensuring its reproducibility on a computer. There we also justify that the algorithm converges in a finite---though random---number of steps.

\subsection{Properties of discrete-time models}

The discrete shell model provides a random final state $\mathbf{u}_{\mathrm{fin}}$ at time $T$ together with intermediate states at the nodes of the space-time lattice.
Two types of nodes can be distinguished, which are indicated by dark green and yellow circles in Fig.~\ref{fig5}, featuring different stages in processing the $B_n^+$ and $B_n^-$ interactions.
At the dark green nodes, the algorithm provides an approximation to the Sabra model, which is first order in $\varepsilon$. 
The yellow (\textsf{T}-shaped) nodes have lower accuracy because the $B_n^+$ and $B_n^-$ interactions are not yet included in the correct proportion; these are auxiliary nodes that can be ignored for the final result.

Figure~\ref{fig6}(a) presents a solution obtained by numerical simulation of the discrete-time model with $\varepsilon = 0.1$, the Reynolds number $\mathrm{Re} = 10^9$ and the remaining parameters set as in Section~\ref{subsec_NSpSt}. Comparing with Fig.~\ref{fig1}(a), we see that the results are very close before the blowup and are followed by similar turbulent dynamics. In Fig.~\ref{fig6}(b) we show the PDFs of the real part $\mathrm{Re}\,u_1(T)$ at time $T = 1.5$ for $\mathrm{Re} = 10^7$ and different $\varepsilon$. These results confirm the convergence to the original (fluctuating Sabra) model as $\varepsilon \to 0$. 

\begin{figure}[tp]
\centering
\includegraphics[width=0.9\textwidth]{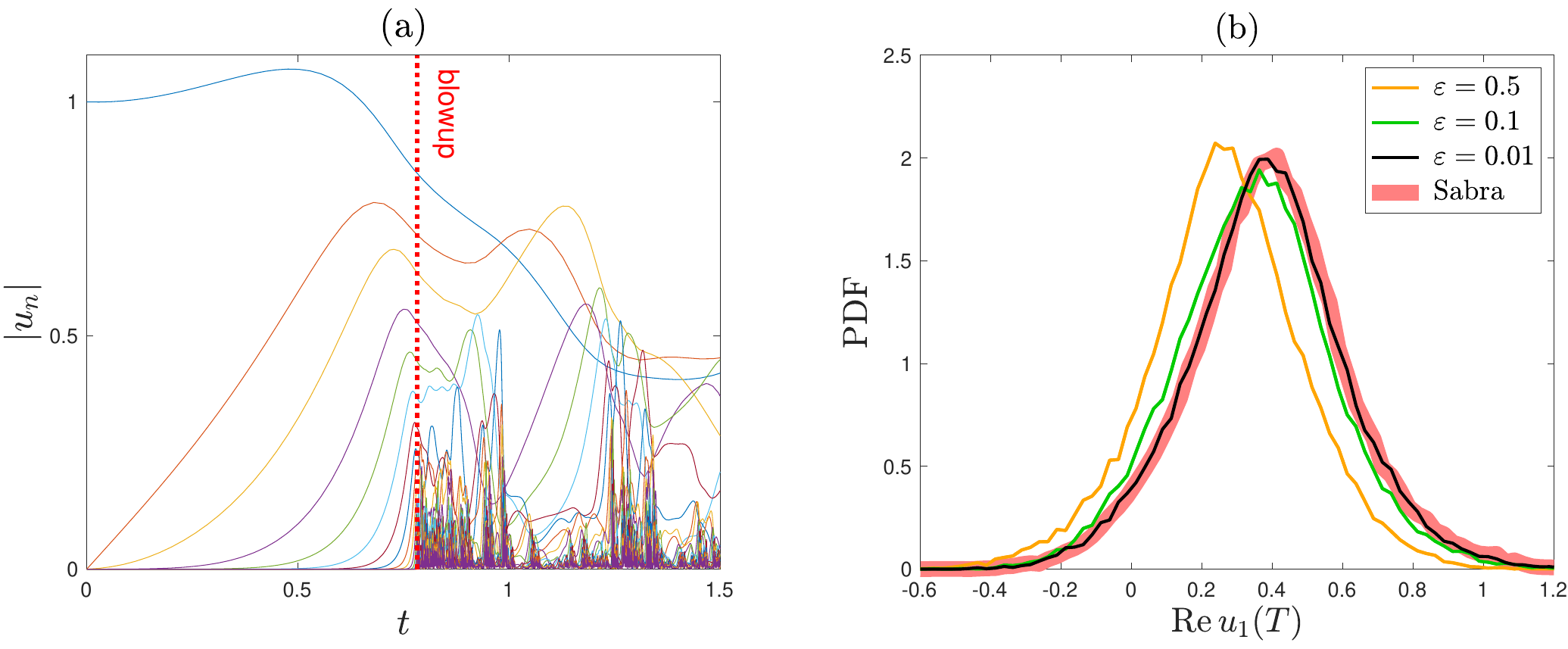}
\caption{(a) Evolution of the discrete-time Sabra model for $\varepsilon = 0.1$ and $\mathrm{Re} = 10^9$. (b) PDFs of $\mathrm{Re}\,u_1(T)$ for $\mathrm{Re} = 10^7$ and different $\varepsilon = 0.5$, $0.1$ and $0.01$ converging to the PDF of the fluctuating Sabra model.}
\label{fig6}
\end{figure}

The specific design of our algorithm allows us to preserve several key properties of the Sabra model.
First, the nonlinear scheme (\ref{eqDTL_b}) preserves both the time and space scale invariance of the ideal system; see Eqs.~(\ref{eq1_ST}) and (\ref{eq1_SS}).
Second, it conserves the energy: $\mathcal{E}[\mathbf{u}] = \mathcal{E}[\mathbf{u}']$.
Such precise conservation of energy is important for the stability of the algorithm, which otherwise may blowup in finite time. 
The energy changes at steps (\ref{eq1_10BC}) as the work done by the boundary conditions, and dissipates due to viscosity at  steps (\ref{eq1_10}). 
Due to scale invariance and energy conservation, the discrete-time Sabra model features intermittency with anomalous scaling analogous to the Navier-Stokes turbulence~\cite{frisch1999turbulence}; see Appendix~\ref{subsecA_2}. 
Thus, in addition to approximating the Sabra model, our discrete-time models are themselves full-fledged toy models of turbulence.

\subsection{Discrete-time LES shell model}

Similarly, one constructs the discrete-time version of the fluctuating LES model (\ref{eq1_5SGS}), which we call the DT-LES model. Here all steps of the algorithm remain precisely the same, except for Eq.~(\ref{eq1_9}) for the viscous and noise terms. The dissipation and noise terms are only present at the shells $N-1$ and $N$. Thus, Eq.~(\ref{eq1_9}) is replaced by
	\begin{equation}
	\label{eq1_9b}
	\frac{du_n}{dt} = -\nu_n u_n + \sigma_n \eta_n, \quad
	\nu_n = D k_n |u'_{n}|, \quad \sigma_n = C k_n^{1/2} |u'_{n}|^{3/2}, 
	\end{equation}
which is used only if $n = N-1$ or $N$.
Its solution is given by the same Eq.~(\ref{eq1_10}) but with new parameters $\nu_n$ and $\sigma_n$. 
Note that Eq.~(\ref{eq1_9b}) preserves the space-scale invariance (\ref{eq1_SSN}) of the fluctuating LES model.

The DT-Sabra and DT-LES models represent only one version of a discrete-time algorithm that respects scale invariance and energy conservation of the original continuous-time models. 
Other possibilities could be considered, in particular, providing higher orders of approximation to the Sabra model.
In this paper we do not explore these possibilities. 
Here we exploit the advantage of discrete-time models in that they allow us to formulate and verify an RG theory for spontaneous stochasticity.

\section{RG relation for DT-LES models} \label{sec4}

In this section we study the dependence of the DT-LES models of the cutoff shell $N$. 
Our goal is to derive an explicit relation between solutions of this model for two successive cutoffs, $N$ and $N+1$. More specifically, we introduce the notion of a flow kernel $\Phi^{(N)}$ that defines transition probabilities at any times and for any initial and boundary conditions. Then we derive the RG operator, $\mathcal{R}: \Phi^{(N)} \mapsto \Phi^{(N+1)}$, which expresses the flow kernel $\Phi^{(N+1)}$ in terms of $\Phi^{(N)}$. 
This derivation is based on the discrete form and scaling symmetries of DT-LES models. 

The central idea is that the RG operator represents the ideal limit $N \to \infty$ as the RG dynamics in the space of flow kernels, thereby linking limiting solutions to RG attractors. 
This theory has previously been developed for turbulence models on self-similar (fractal) lattices~\cite{mailybaev2023spontaneous,mailybaev2025rg}, and here we extend it to more sophisticated shell models.

\subsection{Flow kernel}

Let us consider the states at synchronized times $\boldsymbol{\tau}_j = (t_j,\ldots,t_j)$, where the index $j = 0,1,2,\ldots$ labels the time steps at the shell $n = 1$; see Figs.~\ref{fig5} and \ref{fig7}. We introduce the notation $\mathbf{s}_j = (\mathbf{u}_j,t_j)$, which collects the values of the state variables and the corresponding time into a single vector. 
Then  $\mathbf{s}_0 = (\mathbf{u}_{\mathrm{ini}},0)$ is the initial state and we denote the final state by $\mathbf{s}_{\mathrm{fin}} = (\mathbf{u}_{\mathrm{fin}},T)$. 

According to the algorithm described in Section \ref{subsec_DTSM}, the relation between two successive states $\mathbf{s}_{j}$ and $\mathbf{s}_{j+1}$ can be expressed as a sequence of two transitions. 
The first transition $\mathbf{s}_j \mapsto \tilde{\mathbf{s}}$ is given by a deterministic map denoted as $\alpha_{T,\mathbf{b}}$. This map is defined by Eq.~(\ref{eq1_10BC}) with $\boldsymbol{\tau} = (t_j,\ldots,t_j)$ and $\Delta \tau_1$ given by Eq.~(\ref{eqDTS_TSt}); see loop arrows in Fig.~\ref{fig7}.
The second transition $\tilde{\mathbf{s}} \mapsto \mathbf{s}_{j+1}$ is stochastic and is defined by the composition of steps (\ref{eqDTS_Tnew}), (\ref{eqDTL_b}) and (\ref{eq1_10}) leading to the next synchronized state with $\boldsymbol{\tau} = (t_{j+1},\ldots,t_{j+1})$; see straight arrows in Fig.~\ref{fig7}. This transition is represented by a Markov kernel $\Phi_{T,\mathbf{b}}^{(N)}(\mathbf{s}_{j+1}|\tilde{\mathbf{s}})$. Mathematically, this kernel is a probability distribution (probability measure) for $\mathbf{s}_{j+1}$ given the state $\tilde{\mathbf{s}}$. Note that both $\alpha_{T,\mathbf{b}}$ and $\Phi_{T,\mathbf{b}}^{(N)}$ depend on $T$ and $\mathbf{b} = (b_{-1},b_0)$, but are independent of $j$. Then, the transition probability for $\mathbf{s}_j \mapsto \mathbf{s}_{j+1}$ is described by the Markov kernel
	\begin{equation}
	\label{eqFK_1}
	\big(\Phi_{T,\mathbf{b}}^{(N)} \circ \mathrm{A}_{T,\mathbf{b}}\big)(\mathbf{s}_{j+1}|\mathbf{s}_j),
	\end{equation}
which is the same for all values of $j$.
Here $\mathrm{A}_{T,\mathbf{b}}(\tilde{\mathbf{s}}|\mathbf{s}_j) = \delta\big(\tilde{\mathbf{s}}-\alpha_{T,\mathbf{b}}(\mathbf{s}_j)\big)$ is the Dirac kernel corresponding to the map $\alpha_{T,\mathbf{b}}$, and we refer to  Appendix~\ref{subsecA_3} for definitions of composition and other operations between probability kernels.

\begin{figure}[tp]
\centering
\includegraphics[width=0.85\textwidth]{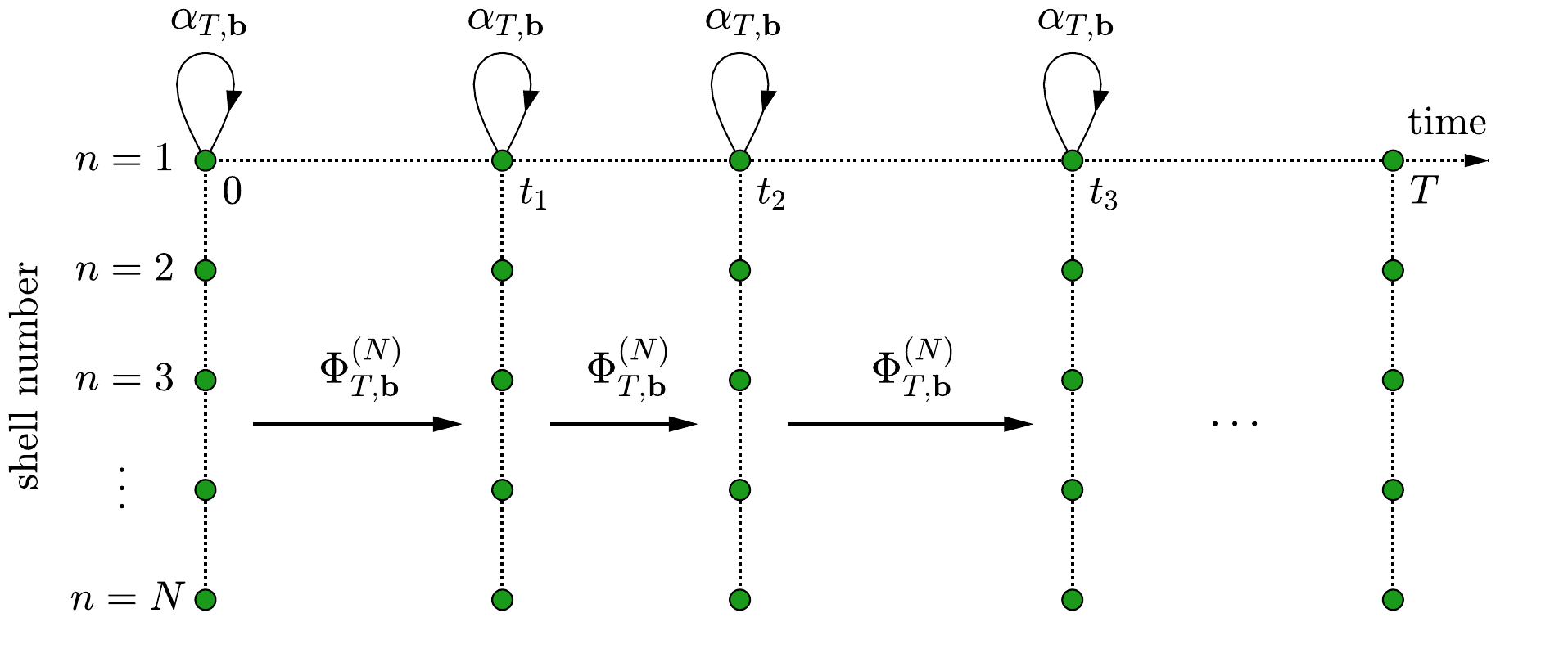}
\caption{Representation of the discrete-time algorithm of Fig.~\ref{fig5} in terms of deterministic maps $\alpha_{T,\mathbf{b}}$ and stochastic transitions governed by the flow kernels $\Phi_{T,\mathbf{b}}^{(N)}$.}
\label{fig7}
\end{figure}

Observe that the transitions $\mathbf{s}_j \mapsto \mathbf{s}_{j+1}$ for different values of $j$ are statistically independent. Hence, the transition probability from $\mathbf{s}_0$ to $\mathbf{s}_m$, which consists of $m$ independent steps, is determined by the probability kernel
	\begin{equation}
	\label{eqFK_1K}
	\big(\Phi_{T,\mathbf{b}}^{(N)} \circ \mathrm{A}_{T,\mathbf{b}}\big)^m(\mathbf{s}_m|\mathbf{s}_0)
	= \Big( \underbrace{\Phi_{T,\mathbf{b}}^{(N)} \circ \mathrm{A}_{T,\mathbf{b}} \circ \cdots \circ 
	\Phi_{T,\mathbf{b}}^{(N)} \circ \mathrm{A}_{T,\mathbf{b}}}_{m\ \textrm{times}} \Big)(\mathbf{s}_m|\mathbf{s}_0).
	\end{equation}
The number of steps $m$ between the initial state $\mathbf{s}_0 = (\mathbf{u}_{\mathrm{ini}},0)$ and the final state $\mathbf{s}_{\mathrm{fin}} = (\mathbf{u}_{\mathrm{fin}},T)$ is a random variable.
Treating the attainment of the final state as a stopping condition, we denote the corresponding kernel by
	\begin{equation}
	\label{eqFK_fin}
	\big(\Phi_{T,\mathbf{b}}^{(N)} \circ \mathrm{A}_{T,\mathbf{b}}\big)^{\mathrm{fin}},
	\end{equation}
with the understanding that a sufficient number of iterations is performed to reach the final time $T$.
Dropping the temporal components of the states $\mathbf{s}_0$ and $\mathbf{s}_{\mathrm{fin}}$, which are now redundant, we obtain the transition probability for $\mathbf{u}_{\mathrm{ini}} \mapsto \mathbf{u}_{\mathrm{fin}}$ in the form of a Markov kernel
	\begin{equation}
	\label{eqFK_1KF}
	\mathrm{P}_{T,\mathbf{b}}^{(N)}(\mathbf{u}_{\mathrm{fin}}|\mathbf{u}_{\mathrm{ini}}) = 
	\big(\Phi_{T,\mathbf{b}}^{(N)} \circ \mathrm{A}_{T,\mathbf{b}}\big)^{\mathrm{fin}}(\mathbf{s}_{\mathrm{fin}}|\mathbf{s}_0).
	\end{equation}

We will use the short notations 
	\begin{equation}
	\label{eqFK_flow}
	\Phi^{(N)} = \{ \Phi_{T,\mathbf{b}}^{(N)}: T \ge 0,\, \mathbf{b} \in \mathbb{C}^2\}, \quad
	\mathrm{P}^{(N)} = \{\mathrm{P}_{T,\mathbf{b}}^{(N)} : T \ge 0,\, \mathbf{b} \in \mathbb{C}^2\}
	\end{equation}
for the families of kernels with arbitrary times and boundary conditions. 
They are related as
	\begin{equation}
	\label{eqFK_1KFb}
	\mathrm{P}^{(N)} = \mathcal{P}[\Phi^{(N)}],
	\end{equation}
where the operator $\mathcal{P}$ denotes the transformation (\ref{eqFK_1KF}). 
By analogy with the theory of dynamical systems, we call $\Phi^{(N)}$ a flow kernel.
We conclude that the flow kernel $\Phi^{(N)}$ defines all transition probabilities $\mathrm{P}^{(N)}$ in the DT-LES model for a given  cutoff $N$.

\subsection{RG relation between the flow kernels}

In this section, we derive the RG operator, which expresses the flow kernel $\Phi^{(N+1)}$ in terms of $\Phi^{(N)}$. 
Since we are dealing with different $N$, it is convenient to think of the states $\mathbf{u} = (u_1,u_2,\ldots)$ as infinite sequences, which are filled by zeros at the shells $n > N$. 
Let us first introduce some deterministic maps that are needed to formulate our main result:
	\begin{equation}
	\label{eqFK_3}
	\begin{array}{ll}
	\sigma_+: (\mathbf{u},t) \mapsto (\mathbf{u}',t),  & \mathbf{u}' = 2(u_2,u_3,\ldots); \\
	\sigma_-: (\mathbf{u},t) \mapsto (\mathbf{u}',t),  & \mathbf{u}' = \frac{1}{2}(0,u_1,u_2,\ldots); \\
	\pi_1: (\mathbf{u},t) \mapsto (\mathbf{u}',0),  & \mathbf{u}' = (u_1,0,0,\ldots).
	\end{array}
	\end{equation}
The first two maps are related to the space-scaling transformation (\ref{eq1_SS}). Together with the projection $\pi_1$ they yields the identity as $\sigma_+ \circ \sigma_- = \pi_1+\sigma_- \circ \sigma_+ = \mathrm{id}$. For each map (\ref{eqFK_3}), we  denote the respective (Dirac) probability kernel by the capital latter $\Sigma_+$, $\Sigma_-$ and $\Pi_1$; see Appendix~\ref{subsecA_3}.
We also define the map
	\begin{equation}
	\label{eqFK_4}
	\xi_{T,\mathbf{b}}: (\mathbf{u},t) \mapsto (\mathbf{u}',t'),  \quad \mathbf{u}' = (u'_1,u'_2,u_3,u_4,\ldots), \quad 
	t' = \min \left( t+\varepsilon T_1,\, T \right),
	\end{equation}
where $u'_1$, $u'_2$ and $T_1$ are given by Eqs.~(\ref{eqDTS_0T}) and (\ref{eqDTL_b}) for $n = 1$. 

\begin{theorem}
\label{th1}
For $N \ge 2$, the flow kernels are related as
	\begin{equation}
	\label{eqFK_6}
	\Phi^{(N+1)} = \mathcal{R}[\Phi^{(N)}],
	\end{equation}
where $\mathcal{R}$ is called the RG operator and is defined by the relations
	\begin{equation}
	\label{eqFK_5}
	\Phi^{(N+1)}_{T,\mathbf{b}} (\mathbf{s}_{\mathrm{new}}|\mathbf{s}) 
	= \left[ \Pi_1 \ast \Sigma_- \circ \big(\Phi^{(N)}_{t',\mathbf{b}'} \big)^{\mathrm{fin}}
	\circ \Sigma_+  \right] \left(\mathbf{s}_{\mathrm{new}}|\mathbf{s}'\right), 
	\quad
	\mathbf{s}' = \xi_{T,\mathbf{b}}(\mathbf{s}), \quad 
	\mathbf{b}' = 2(b_0,u'_1).
	\end{equation}
As before, $\big(\Phi^{(N)}_{t',\mathbf{b}'} \big)^{\mathrm{fin}}$ denotes the composition $\Phi^{(N)}_{t',\mathbf{b}'} \circ \cdots \circ \Phi^{(N)}_{t',\mathbf{b}'}$, where the number of kernels is sufficient to reach the final time $t'$.
\end{theorem}

\begin{proof}

\begin{figure}[tp]
\centering
\includegraphics[width=0.9\textwidth]{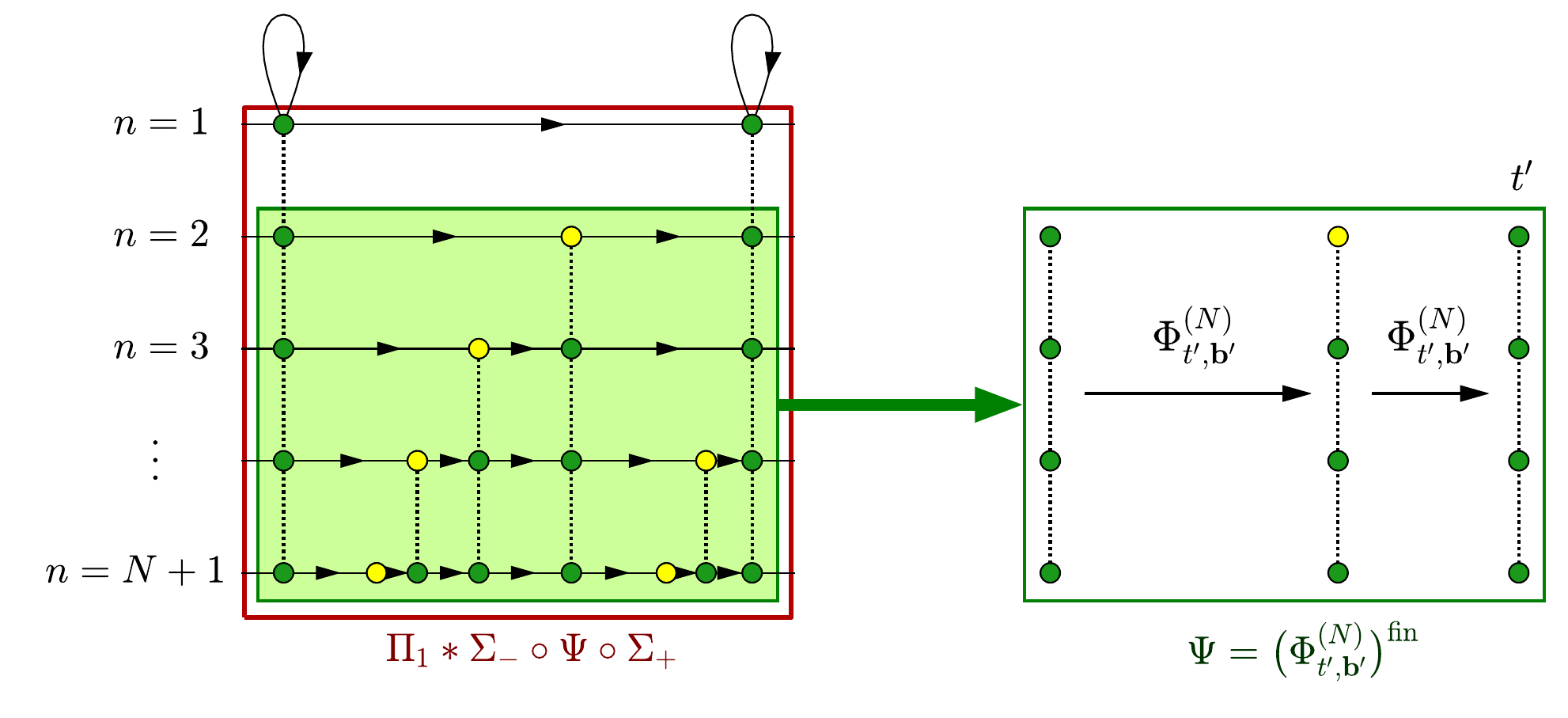}
\caption{Using the self-similar structure of DT-LES model with the cutoff $N+1$, we represent the dynamics at shells $n \ge 2$ in terms of the flow kernel corresponding to the cutoff $N$.}
\label{fig8}
\end{figure}

Consider the kernel $\Phi_{T,\mathbf{b}}^{(N+1)}(\mathbf{s}_{\mathrm{new}}|\mathbf{s})$ from Eq.~(\ref{eqFK_1}) for the increased cutoff $N+1$. 
It describes the probability transition $\mathbf{s} \mapsto \mathbf{s}_{\mathrm{new}}$ corresponding to the algorithm steps (\ref{eqDTS_Tnew}), (\ref{eqDTL_b}) and (\ref{eq1_10}); see Fig.~\ref{fig5}. 
The first step, which corresponds to $n = 1$, yields the state $\mathbf{s}'$ given by the map (\ref{eqFK_4}); recall that the viscous and noise step is not present because $N+1 \ge 3$; see Eq.~(\ref{eq1_9b}). The remaining steps affect only the components $n \ge 2$ and, therefore, can be represented by a kernel $\Psi$ describing the transition of the shifted state $\sigma_+(\mathbf{s}') \mapsto \mathbf{s}''$; see Fig.~\ref{fig8}. This yields $\mathbf{s}_{\mathrm{new}} = \pi_1(\mathbf{s}')+\sigma_- (\mathbf{s}'')$, where $\sigma_-$ is the backward shift and $\pi_1$ restores the first component. For the probability kernels, this yields 
	\begin{equation}
	\label{eqFK_5P1}
	\Phi^{(N+1)}_{T,\mathbf{b}} (\mathbf{s}_{\mathrm{new}}|\mathbf{s}) 
	= \left[ \Pi_1 \ast \Sigma_- \circ \Psi
	\circ \Sigma_+  \right] \left(\mathbf{s}_{\mathrm{new}}|\mathbf{s}'\right),
	\end{equation}
where the sum of states is replaced by a convolution of kernels; see  Appendix~\ref{subsecA_3}. Finally, as a consequence of the scaling symmetry (\ref{eq1_SSN}), one can see that the kernel $\Psi$ is given by essentially the same algorithm; see Fig.~\ref{fig8}. This algorithm is formulated for the new final time $t'$, the shifted boundary state $\mathbf{b}'$, the cutoff $N$, and also the loop steps (\ref{eq1_10BC}) must be skipped.
Thus, similarly to the relation (\ref{eqFK_1KF}) we obtain $\Psi = \big(\Phi_{t',\mathbf{b}'}^{(N)}\big)^{\mathrm{fin}}$. 
Combining this with Eq.~(\ref{eqFK_5P1}) proves the theorem.
\end{proof}

The RG relations of Theorem~\ref{th1} extend to a wide class of DT-LES models. 
One can verify that the same relations are valid for any dissipative and noise terms whose dependence on the cutoff $N$ has scaling symmetry (\ref{eq1_SSN}). In such models, the flow kernels for different $N$ are related by the same RG operator (\ref{eqFK_6}).

\section{Spontaneous stochasticity from a fixed-point RG attractor}
\label{secRGFP}

\subsection{Ideal limit as a dynamical system}

The main consequence of the RG relation (\ref{eqFK_6}) is that it represents the ideal limit of $N \to \infty$ as a dynamical system in the space of flow kernels. The RG operator $\mathcal{R}$, which determines this dynamics, depends only on the ideal system, i.e., on the nonlinear interactions $B_1^+$ and $B_1^-$ determining the map (\ref{eqFK_4}) through Eq.~(\ref{eqDTL_b}) with $n = 1$. 
This operator does not depend on the specific form of dissipation and noise at small scales, in particular, on the values of the coefficients $D$ and $C$ in case~(\ref{eq1_9b}). 
Which particular small-scale regularization is used is ``encoded'' in the flow kernel $\Phi^{(N)}$ and is passed on to the next flow kernel $\Phi^{(N+1)}$ due to scaling symmetry (\ref{eq1_SSN}).
In other words, the ideal model determines the RG operator, while the choice of regularization determines the initial flow kernel $\Phi^{(N_0)}$ for the RG dynamics.

By representing the ideal limit as the RG dynamics, we also identify the probability distributions obtained in this limit with an attractor of this dynamics. 
We now formulate the concept of a fixed-point RG attractor, which is based on the assumption that the classical theory of dynamical systems is applicable; see e.g. \cite{arnold1992ordinary}. 
Of course, this does not follow automatically, since the RG dynamics occurs in an infinite-dimensional space of flow kernels, but our numerical results will strongly support this assumption. 
The existence and stability of the RG attractor provides a theoretical explanation for spontaneous stochasticity, including its universality property and convergence process.

\subsection{Fixed-point RG attractor}

A flow kernel $\Phi^{\infty}$ is the fixed-point RG attractor if it satisfies the fixed-point relation $\mathcal{R}[\Phi^{\infty}] = \Phi^{\infty}$ and has the convergence property
	\begin{equation}
	\label{eqSFP_1}
	 \lim_{N \to \infty } \mathcal{R}^{N}[\Phi] = \Phi^{\infty}, \quad 
	\Phi \in \mathcal{B}(\Phi^{\infty}).
	\end{equation}
Here $\mathcal{B}(\Phi^{\infty})$ denotes the basin of attraction in the space of flow kernels, meaning that $\mathcal{B}(\Phi^{\infty})$ contains an open neighborhood of the fixed point $\Phi^{\infty}$ in an appropriate mathematical sense.
Using relation (\ref{eqFK_1KFb}), we write the ideal limit for the family of transition probabilities as
	\begin{equation}
	\label{eqSFP_1P}
	\lim_{N \to \infty } \mathcal{P}\big[ \mathcal{R}^{N}[\Phi] \big] = \mathrm{P}^{\infty}, \quad 
	\mathrm{P}^{\infty} = \mathcal{P}[\Phi^{\infty}].
	\end{equation}
These limits apply to the DT-LES models if the respective flow kernels belong to the basin of attraction.

The condition of RG attractor is conveniently expressed in terms of observables considered as functions $O: \mathbf{u}_{\mathrm{fin}} \mapsto \mathbb{R}$ of the final state $\mathbf{u}_{\mathrm{fin}}$. 
Given initial conditions $\mathbf{a}$, boundary conditions $\mathbf{b}$, and finite time $T$, we define the mean values as
	\begin{equation}
	\label{eqRGN_1b}
	\langle O \rangle^{(N)}_{T,\mathbf{b},\mathbf{a}} = \int O[\mathbf{u}] \mathrm{P}^{(N)}_{T,\mathbf{b}}(d\mathbf{u}|\mathbf{a}), \quad
	\langle O \rangle^\infty_{T,\mathbf{b},\mathbf{a}} = \int O[\mathbf{u}] \mathrm{P}^{\infty}_{T,\mathbf{b}}(d\mathbf{u}|\mathbf{a}),
	\end{equation}
which are expectations with respect to the respective probability distributions. 
Then the convergence to the RG attractor implies 
	\begin{equation}
	\label{eqRGN_1}
	\lim_{N \to \infty } \langle O \rangle^{(N)}_{T,\mathbf{b},\mathbf{a}} 
	= \langle O \rangle^\infty_{T,\mathbf{b},\mathbf{a}}.
	\end{equation}

We can now relate the existence of the RG attractor to the four properties of spontaneous stochasticity listed Section~\ref{subsecSpSt_A}. ($i$) The RG attractor implies naturally the convergence for respective probability distributions. 
($ii$) The stochasticity property follows if the limiting  family of transition probabilities $\mathrm{P}^{\infty}$ is not Dirac. 
($iii$) The spontaneity property requires that realizations of the limiting process solve the ideal model. This is true because the dynamics at each given scale is governed by ideal couplings for $N$ sufficiently large.
($iv$) Finally, universality with respect to the formulation of the ideal limit is a consequence of the attracting condition (\ref{eqSFP_1P}). 
This universality refers to all types of small-scale dissipation and noise, for which the respective flow kernels: (a) satisfy the RG relation $\Phi^{(N+1)} = \mathcal{R}[\Phi^{(N)}]$ and (b) belong to the basin of attraction. 

We remark that the RG formalism can be introduced at the level of the continuous-time model, similarly to~\cite{mailybaev2024rg}, as long as the model can be formulated as a finite system of SDEs. This condition, however, does not hold at the fixed point, and therefore the RG operator is not well defined there. Consequently, the linearized RG problem developed below also loses a clear mathematical meaning. These difficulties are resolved at the level of the discrete-time model.

\subsection{Linearization and stability of a fixed-point RG attractor}

In this subsection we assume that classical stability theory is applicable to RG dynamics. Then the RG operator can be linearized near its fixed point $\Phi^{\infty}$ such that
	\begin{equation}
	\label{eqSFP_2}
	\mathcal{R}[\Phi] \approx \Phi^{\infty}+\delta\mathcal{R}[\Psi], \quad \Psi = \Phi-\Phi^\infty,
	\end{equation}
where $\delta\mathcal{R}$ is a linear operator acting on perturbations $\Psi$ of the flow kernel. Note that $\Psi$ denotes a family of kernels $\Psi_{T,\mathbf{b}}(\mathbf{s}_{\mathrm{new}}|\mathbf{s})$, which are signed measures with respect to $\mathbf{s}_{\mathrm{new}}$. Substituting Eq.~(\ref{eqSFP_2}) into the RG relation (\ref{eqFK_6}) and denoting $\Psi^{(N)} = \Phi^{(N)}-\Phi^{\infty}$, we obtain the linearized dynamics 
	\begin{equation}
	\label{eqSFP_3}
	 \Psi^{(N+1)} = \delta\mathcal{R}[\Psi^{(N)}].
	\end{equation}
Following classical stability analysis, we assume that the linearized RG dynamics (\ref{eqSFP_3}) in the ideal limit $N \to \infty$  approaches asymptotically the eigenmode solution
	\begin{equation}
	\label{eqSFP_4}
	\Psi^{(N)} \approx c \rho^N \Omega.
	\end{equation}
Here $\rho$ is a leading (largest absolute value) real eigenvalue and $\Omega$ a corresponding
eigenvector of the eigenvalue problem
	\begin{equation}
	\label{eqSFP_5}
	\delta\mathcal{R}[\Omega] = \rho \Omega.
	\end{equation}
The eigenvector represents a collection of (signed) kernels $\Omega = \{\Omega_{T,\mathbf{b}} : T \ge 0,\, \mathbf{b} \in \mathbb{C}^2\}$. 
	
The stability condition is formulated as $|\rho| < 1$, in which case the correction (\ref{eqSFP_4}) decays as $N \to \infty$. 
It follows that both $\rho$ and $\Omega$ are universal, i.e. they do not depend on the regularization parameters (small-scale dissipation and noise). The prefactor $c \in \mathbb{R}$ is the only part of the expression (\ref{eqSFP_4}) that is not universal.

Similarly, we linearize the operator (\ref{eqFK_1KFb}) as
	\begin{equation}
	\label{eqSFP_6a}
	\mathcal{P}[\Phi] \approx \mathrm{P}^{\infty}+\delta\mathcal{P}[\Psi], \quad 
	\mathrm{P}^{\infty} = \mathcal{P}[\Phi^\infty], \quad \Psi = \Phi-\Phi^\infty.
	\end{equation}
Then, using Eqs.~(\ref{eqFK_1KFb}), (\ref{eqSFP_4}) and (\ref{eqSFP_6a}), we obtain the asymptotic form
	\begin{equation}
	\label{eqSFP_6}
	\mathrm{P}^{(N)} \approx 
	\mathrm{P}^\infty+c\rho^N  \mathrm{Q}, \quad
	\mathrm{Q} = \delta\mathcal{P}[\Omega],
	\end{equation}
describing the convergence of transitional probabilities in the ideal limit. 

Asymptotic expressions (\ref{eqSFP_4}) and (\ref{eqSFP_6}) are written for a real eigenmode. 
However, the eigenvalue $\rho$ and the respective eigenvector $\Omega$ in Eq.~(\ref{eqSFP_5}) can be complex. In this case the asymptotic expressions (\ref{eqSFP_4}) and (\ref{eqSFP_6}) take the form (see e.g. \cite{arnold1992ordinary})
	\begin{equation}
	\label{eqSFP_8}
	\Phi^{(N)} \approx \Phi^{\infty}+\mathrm{Re}\big( c \rho^N\Omega \big), \quad
	\mathrm{P}^{(N)} \approx \mathrm{P}^\infty
	+ \mathrm{Re}\big( c \rho^N \mathrm{Q} \big), 
	\quad
	\mathrm{Q} = \delta\mathcal{P}[\Omega],
	\end{equation}
where the non-universal prefactor $c$ is also complex.
Relations (\ref{eqSFP_8}) lead to a correction term in the limit~(\ref{eqRGN_1}) as
	\begin{equation}
	\label{eqRGN_2}
	\langle O \rangle^{(N)}_{T,\mathbf{b},\mathbf{a}} 
	\approx \langle O \rangle^{\infty}_{T,\mathbf{b},\mathbf{a}}
	+\mathrm{Re}\big( c\rho^N q \big), \quad
	q = \int O[\mathbf{u}] \mathrm{Q}_{T,\mathbf{b}}(d\mathbf{u}|\mathbf{a}),
	\end{equation}
where the second expression integrates the observable $O$ with respect to the eigenvector kernel $\mathrm{Q}_{T,\mathbf{b}}$ similarly to Eq.~(\ref{eqRGN_1b}).
Recall that the RG operator and, hence, the eigenmode depend on the ideal system only. Therefore, the quantities in the asymptotic relation~(\ref{eqRGN_2}) have different forms of universality: (a) The eigenvalue $\rho$ is fully universal in the sense that it does not depend on the initial and boundary conditions, the final time, small-scale regularization, or the observable. (b) The integral $q$ depends on the initial and boundary conditions, the final time and on the observable, but is independent of the regularization. (c) Finally, the prefactor $c$ depends on the regularization, but does not depend on the initial and boundary conditions, the final time, and the observable.
Recall that the independence of regularization is assumed here within the class of models to which the RG theory is applicable and whose flow kernels belong to the basin of the RG attractor.

\subsection{Numerical verification of the universal RG eigenvalue}

The RG attractor hypothesis explains the convergence to spontaneously stochastic solutions in the ideal limit, which we already observed numerically. In this section, we focus on the predictions of RG theory that are given by the complex eigenmode correction (\ref{eqRGN_2}). This is a new universality property, and its numerical confirmation provides additional support for our RG approach.
We perform this verification by simulating the DT-LES model with the discretization parameter $\varepsilon = 0.1$. We use dissipative and noise terms (\ref{eq1_9b}) with $D = 1$ and $C = 0.5$. The final time, initial and boundary conditions are as in Section~\ref{subsec_NSpSt}.

Consider the observable $O = |u_m|^p$ for some given shell $m$ and order $p$.
Then the expectation $\langle O \rangle^{(N)}_{T,\mathbf{b},\mathbf{a}} = \langle |u_m(T)|^p \rangle$ is the average over realizations of small-scale noise for a given finite time $T$, initial and boundary conditions $\mathbf{a}$ and $\mathbf{b}$, and cutoff shell $N$. Let us write relation (\ref{eqRGN_2}) in the form
	\begin{equation}
	\label{eqRGN_2b}
	\langle O \rangle^{(N)}_{T,\mathbf{b},\mathbf{a}} 
	\approx \langle O \rangle^{\infty}_{T,\mathbf{b},\mathbf{a}}
	+C |\rho|^N \cos (N\omega+\phi),
	\end{equation}
where we expressed $\rho = |\rho|e^{i\omega}$ and denoted the absolute value $C = |cq|$ and phase $\phi = \arg (cq)$. 
This asymptotic relation is verified in Fig.~\ref{figN1}, where the circles show the numerical averages $\big\langle |u_m(T)|^p \big\rangle$ for $m = 1,2$ and $p = 1,2$. 
By fitting this data with the RG eigenmode expression (\ref{eqRGN_2b}), where $\rho$ must be the same for different observables, we estimated the eigenvalue as
	\begin{equation}
	\label{eqRGN_2d}
	|\rho| = 0.84\pm 0.02, \quad \omega = \arg\rho = 2.28 \pm 0.01.
	\end{equation}
This fitting is obtained by minimizing a suitable norm of the difference between the left- and right-hand sides of Eq.~(\ref{eqRGN_2b}) with respect to the unknown parameters $\langle O \rangle^{\infty}_{T,\mathbf{b},\mathbf{a}}$, $\rho$, $C$ and $\phi$.
The error bounds in Eq.~(\ref{eqRGN_2d}) correspond to variations across different measurements. Solid lines in Fig.~\ref{figN1} show the right-hand sides of Eq.~(\ref{eqRGN_2b}) with the eigenvalue (\ref{eqRGN_2d}) and optimized values for $C$ and $\phi$. 
In all cases, the measured ensemble averages show very good agreement with the RG prediction for large $N$.
It is noteworthy that $|\rho|$ in Eq.~(\ref{eqRGN_2d}) is close to unity, which explains the slow convergence in the ideal limit. 

\begin{figure}[tp]
\centering
\includegraphics[width=0.9\textwidth]{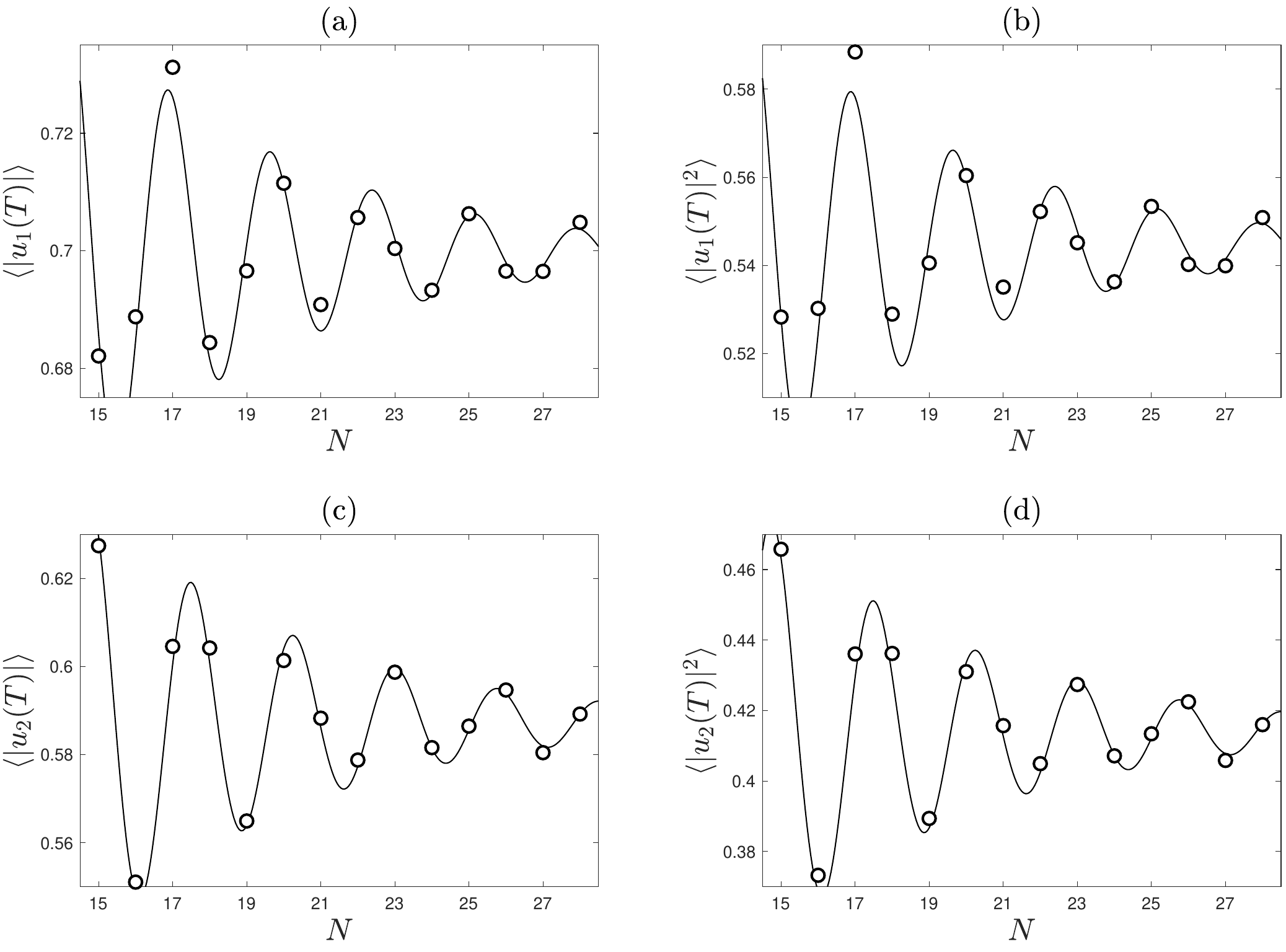}
\caption{The moments (a) $\big\langle |u_1(T)| \big\rangle$, (b) $\big\langle |u_1(T)|^2 \big\rangle$, (c) $\big\langle |u_2(T)| \big\rangle$ and (d) $\big\langle |u_2(T)|^2 \big\rangle$ computed numerically (circles) are compared with their asymptotic RG predictions (lines). The averages are estimated using $10^5$ random samples for each $N$, in which case the statistical errors are of the order of the circle size.}
\label{figN1}
\end{figure}

Finally, we verify numerically the universality of the RG eigenvalue. 
We do this by considering a different regularization with different initial and boundary conditions. First, instead of the large-scale initial state $\mathbf{a} = (1,0,0,\ldots)$, we consider the ``rough'' initial condition in the form of geometric progression $\mathbf{a} = (z,z^2,z^3,\ldots)$, where $z = 2^{-h}\exp(3i)$ with the H\"older exponent $h = 0.3$. We use the same boundary conditions as before with the final time $T = 1$. A typical dynamics is presented in Fig.~\ref{figN5}(a). The moments $\big\langle |u_1(T)| \big\rangle$ and $\big\langle |u_1(T)|^2 \big\rangle$ computed for different $N$ are shown in Figs.~\ref{figN5}(b,c). The solid lines present the asymptotic relations (\ref{eqRGN_2b}), where we used the same RG eigenvalue (\ref{eqRGN_2d}) and tuned prefactors and phases. 

\begin{figure}[tp]
\centering
\includegraphics[width=1\textwidth]{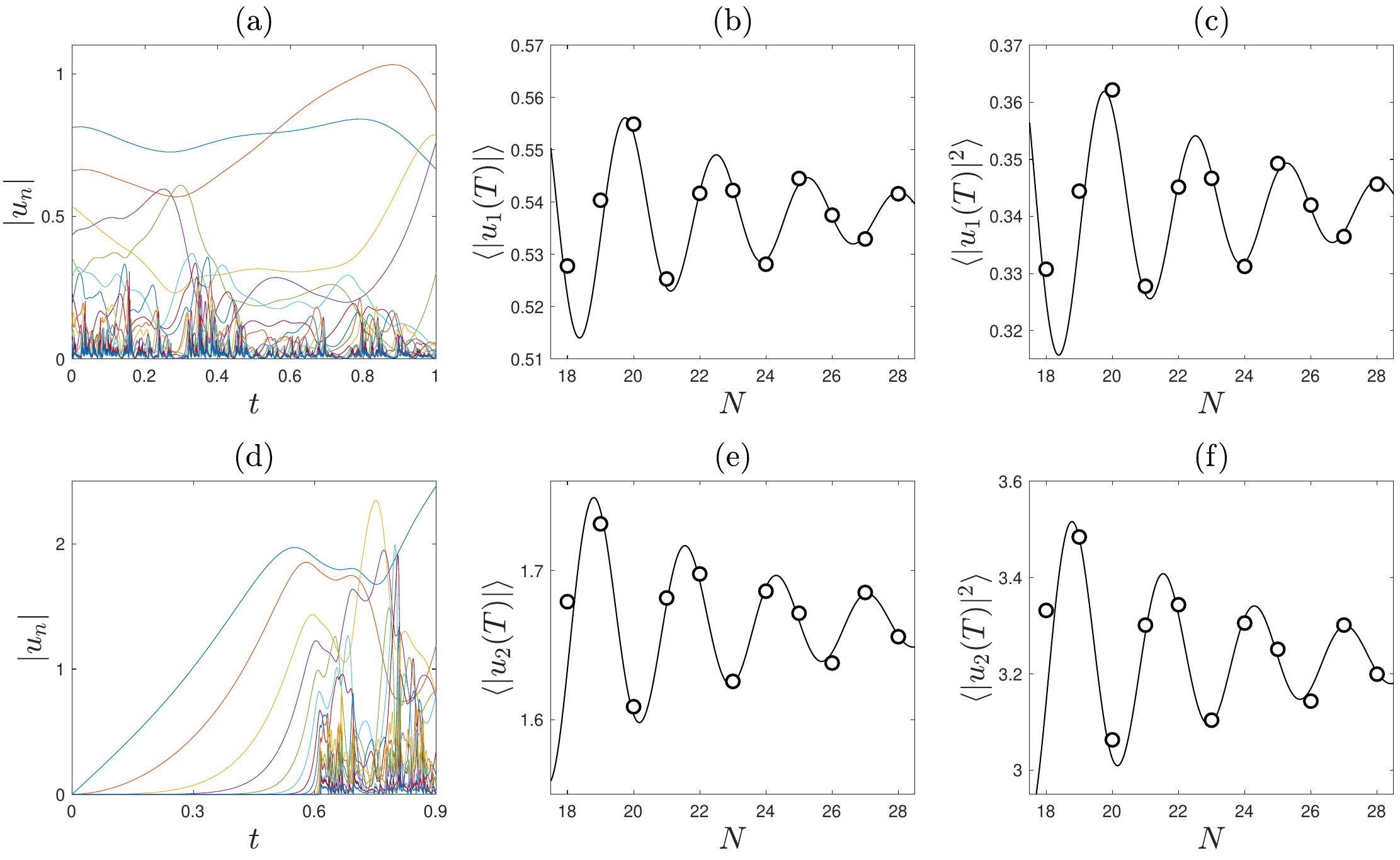}
\caption{(a) Dynamics of the LES model for rough initial data. (b,c) Dependence on the moments $\big\langle |u_1(T)| \big\rangle$ and $\big\langle |u_1(T)|^2 \big\rangle$ on $N$ (empty dots) fitted with the universal RG eigenmode prediction (solid lines). Statistics used $10^5$ simulations, and statistical errors are comparable with the dot size. (d-f) Similar results for the different small-scale regularization (\ref{eq1_5SGSmod}) for zero initial condition.}
\label{figN5}
\end{figure}

Finally, let us consider a different regularization. We modify the LES model (\ref{eq1_5SGS}) by introducing the dissipation and noise at the last three shells as
	\begin{equation}
	\label{eq1_5SGSmod}
	\frac{du_n}{dt} = \left\{ \begin{array}{ll}
	B_n, & n = 1,2,\ldots,N-3; \\
	B_n-D k_n |u_n| u_n+ C k_n^{1/2} |u_n|^{3/2} \eta_n,& n = N-2,\ldots,N;
	\end{array}\right.
	\end{equation}
with the parameters $D = 1$ and $C = 0.1$. The DL-LES model in Eq.~(\ref{eq1_9b}) is modified respectively. For the new boundary conditions we take $b_{-1} = 2$ and $b_0 = 3$, and choose zero initial conditions. This system has blowup time $t_b \approx 0.61$ and we set the final time at $T = 0.9$; see Fig.~\ref{figN5}(d). The averages $\big\langle |u_2(T)| \big\rangle$ and $\big\langle |u_2(T)|^2 \big\rangle$ are presented in Fig.~\ref{figN5}(e,f). Solid lines show the RG asymptotics with the same RG eigenvalue (\ref{eqRGN_2d}). 
In all cases, we observe a good match of the RG eigenmode approximation with the numerical results for large $N$, confirming the universality of the RG eigenvalue.

\section{RG approach to the fluctuating Sabra model} \label{sec5}

The RG approach developed in the previous sections was based on the key property that regularized models preserve the scaling symmetries of the ideal model: the time scaling symmetry (\ref{eq1_ST}) and the extended form of the space scaling symmetry (\ref{eq1_SSN}). 
Two examples are the LES models given by Eqs.~(\ref{eq1_5SGS}) and (\ref{eq1_5SGSmod}). 
We call such symmetry-preserving regularizations canonical. It is the self-similarity of the dissipative and noise terms in the canonical regularization that leads to the existence of the RG operator $\mathcal{R}$ and, thus, to a representation of the ideal limit $N \to \infty$ by the RG dynamics of flow kernels $\Phi^{(N)}$. 

However, from a physical point of view, dissipation and noise mechanisms are not required to preserve  symmetries of the ideal model, and in general they do not. 
For example, these symmetries are broken by both the viscous and noise terms in the fluctuating Sabra model (\ref{eq1_5}). 
In \cite{mailybaev2024rg} it was suggested that the effect of a specific non-canonical regularization can be understood if the model can be expressed in terms of a family of auxiliary canonical regularizations. 
We now show how this idea is applied to the fluctuating Sabra model. We will use a continuous-time formulation of the model, given that it can be approximated by discrete-time models to which RG theory is applicable.

The idea is to consider a family of auxiliary shell models having the form
	\begin{equation}
	\label{eqRGV_1}
	\frac{du_n}{dt} = B_n-\frac{Dk_n^2}{k_*^2 \tau_*} u_n + \frac{Ck_n}{k_*^2 \tau_*^{3/2}} \eta_n, 
	\quad n = 1,\ldots,N,
	\end{equation}
where $D$ and $C$ are fixed parameters describing the intensity of dissipation and noise and 
	\begin{equation}
	\label{eqRGV_1KT}
	k_* = 2^{N-M}, \quad
	\tau_* = \bigg( \sum_{j = -1}^{N-M} k_j^2|u_j|^2 \bigg)^{-1/2}
	\end{equation}
with a fixed integer parameter $M$. 
The dissipation and noise terms in Eq.~(\ref{eqRGV_1}) can be seen as characterized by the effective wave number $k_*$ and the turnover time $\tau_*$. 
The regularizations (\ref{eqRGV_1}) are canonical, i.e. they possess both symmetries (\ref{eq1_ST}) and (\ref{eq1_SSN}). One can check that discrete-time versions of models (\ref{eqRGV_1}) yield the same RG operator from Theorem~\ref{th1}: if $\Phi^{(N)}$ is the flow kernel for the model (\ref{eqRGV_1}), then $\Phi^{(N+1)} = \mathcal{R}[\Phi^{(N)}]$ is the flow kernel for the increased cutoff shell $N+1$. Strictly speaking, this derivation requires neglecting the dissipative and noise terms at shell $n = 1$, which is valid in the inviscid limit $N \to \infty$.
Hence, the ideal limit $N \to \infty$ for each auxiliary model (\ref{eqRGV_1}) is spontaneously stochastic and universal. We confirm this numerically in Fig.~\ref{fig4} for the parameters $D = 1$, $C = 0.01$, $M = 4$ and $N = 24$.

The fluctuating Sabra model (\ref{eq1_5}) follows from auxiliary model (\ref{eqRGV_1}) by choosing the non-stationary parameters $D(t)$ and $C(t)$ in the form
	\begin{equation}
	\label{eqRGV_2}
	D = \mathrm{Re}^{-1} k_*^2 \tau_* , \quad
	C = \sqrt{\Theta} \, k_*^{2} \tau_*^{3/2}.
	\end{equation}
When $\mathrm{Re}^{-1}$ and $\sqrt{\Theta}$ are small, for example in the ideal limit formulated in (\ref{eq1_4}), one can choose the time-dependent $M(t)$ such that both $D(t) \lesssim 1$ and $C(t) \lesssim 1$.
This is illustrated in Fig.~\ref{figS1} for a specific realization of the Sabra model with $\mathrm{Re} = 10^{9}$ from Fig.~\ref{fig1}(a). The correspondence (\ref{eqRGV_2}) means that, at every time instant, the Sabra model behaves as the auxiliary model with specifically chosen parameters. 
The RG theory predicts that the dynamics of the latter is approximated by the universal spontaneously stochastic process. Therefore, the same process approximates the dynamics of the fluctuating Sabra model.

\begin{figure}[tp]
\centering
\includegraphics[width=1\textwidth]{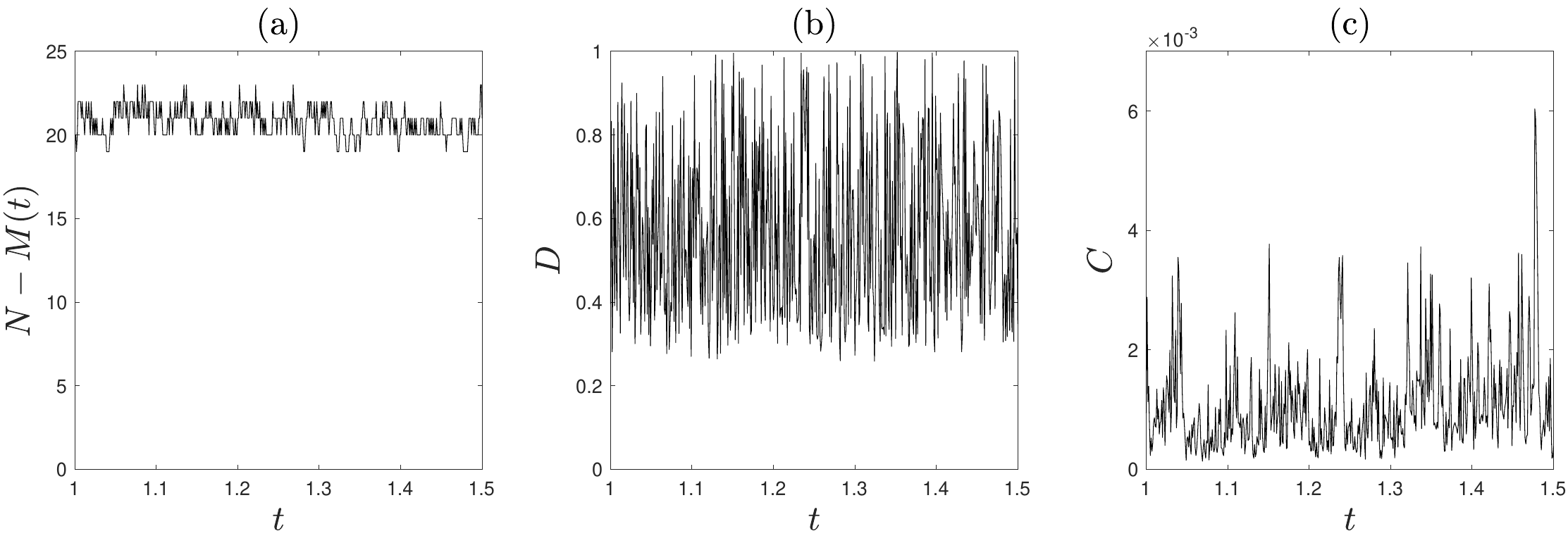}
\caption{Oscillations of the auxiliary model parameters $N-M(t)$, $D(t)$ and $C(t)$.}
\label{figS1}
\end{figure}

The presented argument points to a general way of explaining the universal property of spontaneous stochasticity. This universality is determined by systems with canonical regularizations to which RG theory is applicable. 
In cases where the regularization is non-canonical, such as in the fluctuating Sabra model, universality can be understood by establishing a connection to auxiliary canonical models. 

\section{Conclusion}

In this paper, a renormalization group (RG) formalism is developed for the phenomenon of spontaneous stochasticity in the fluctuating (Sabra) shell model of turbulence. 
For finite viscosities and amplitudes of noise, this model is described by stochastic differential equations, and the solution of the initial value problem is a stochastic process. 
Spontaneous stochasticity means that these solutions converge in an ideal limit in which viscosity and noise simultaneously tend to zero, and the limit remains random. 
The spontaneous nature of stochasticity reflects the fact that the limit process solves the initial value problem for a formally deterministic ideal model. 
Moreover, this limit is universal, i.e. it does not depend on the specific form of the dissipation mechanism and noise. 
In our theory, these properties are associated with the fixed point attractor of the RG operator.

The RG theory is formulated for discrete-time shell models. Such models are developed as approximations of the continuous-time Sabra model, converging to the latter as the approximation parameter $\varepsilon$ tends to zero. These discrete models are built on a dynamic space-time lattice and, from a numerical point of view, represent an adaptive finite-difference scheme. 
The key point is that these approximations exactly preserve scale symmetries and energy balance. 
It appears possible to extend such models to conserve other invariants, such as helicity; however, this does not seem to be important in the context of the present study.

The RG theory applies to class of models, in which the regularization (dissipation and noise) terms are scale invariant. The ideal limit corresponds to $N \to \infty$, where $N$ is a regularization (cutoff) parameter. For each $N$, the stochastic evolution is described by a flow kernel $\Phi^{(N)}$, which is a family of transition probabilities (Markov kernels) for all initial and boundary conditions. We derive the RG operator $\mathcal{R}$, which expresses the flow kernel $\Phi^{(N+1)}$ in terms of $\Phi^{(N)}$. This operator defines a dynamical system in the space of flow kernels and associates the ideal limit with an attractor of the RG dynamics. Here the dynamics (operator $\mathcal{R}$) depends on nonlinear couplings of the ideal system, but does not depend on the regularization (dissipation and noise). 
A specific regularization is introduced into the RG dynamics through the initial condition, i.e. through the initial flow kernel $\Phi^{(N_0)}$.

We conjecture that the RG dynamics has a fixed-point attractor. This property, which we verify with high accuracy using numerical simulations, naturally explains the observed convergence and universality of spontaneous stochasticity. Moreover, it makes a new prediction: convergence in the ideal limit is determined by the leading eigenvalue $\rho$ of the linearized RG operator. Numerical simulations show that $\rho \approx 0.84 \exp(2.28i)$ is a complex number that controls both the overall exponential decay and the oscillations in the ideal limit. This RG eigenvalue is universal: it does not depend on the initial and boundary conditions, as well as on the type of dissipation and noise.

This RG theory is constrained by scale-invariant forms of dissipation and noise, which we call canonical regularizations. Physical regularizations originating from the viscous mechanism and molecular noise are not canonical. We argue, however, that the convergence, spontaneous stochasticity and universality of  the ideal limit follow from the same RG theory. For this purpose, we introduce a family of auxiliary models with canonical regularizations, which coincide with the Sabra model at every instant of time, given a proper (time-dependent) choice of model parameters. Thus, the properties of the fluctuating Sabra model in the ideal limit can be derived from the properties of the auxiliary models, where the latter are described by the RG theory. 

The RG approach to Eulerian spontaneous stochasticity was earlier developed for discrete models on a self-similar (fractal) lattice~\cite{mailybaev2023spontaneous,mailybaev2025rg}. Our current work demonstrates that the same approach works for systems with continuous time such as the Sabra shell model. This is achieved by using a dynamic multiscale space-time lattice. The important question is whether this RG theory can be extended further to the fluctuating Navier-Stokes system? The obstacle for this development is the continuity of physical space and Galilean invariance. The other interesting direction of research is the study of spontaneous stochasticity in general, i.e. without a necessary connection to the Navier-Stokes system. 
To this end, the discrete-time symmetry-preserving models of this paper facilitate a mathematical formulation of the RG theory that is precise and convenient for numerical analysis.

Finally, let us remark on another interesting role of noise by contrasting deterministic and stochastic regularizations. Consider deterministically regularized models, for example by setting $C=0$ in the LES shell model~(\ref{eq1_5SGS}). One can then define a deterministic version of the RG operator acting on flow maps~\cite{mailybaev2024rg}. If such deterministic RG dynamics is chaotic in a classical sense, its invariant probability measure would yield a stochastic description in the limit $N\to\infty$. A natural question is whether this description coincides with the stochastic RG fixed point in Eq.~(\ref{eqSFP_1}).
Although the answer may be positive, this analogy with chaos is not straightforward. First, the deterministic RG dynamics exhibits superexponential separation of trajectories~\cite{mailybaev2024rg}, in contrast to the exponential growth characteristic of classical chaos. More importantly, probability evolution in the deterministic case is linear, as it is given by the pushforward of a deterministic map, whereas the stochastic RG operator in Theorem~\ref{th1} is nonlinear, allowing, for example, periodic stochastic attractors~\cite{mailybaev2025rg}. These differences suggest a fundamental distinction between spontaneous stochasticity and deterministic chaos and point to an interesting direction for future research.

\section{Appendix}

\subsection{Implicit algorithm step} \label{subsec_A1}

Here we solve Eq.~(\ref{eqDTL_b}) representing the implicit finite-difference relation for the nonlinear term. First, let us prove that this transformation preserves the energy $\mathcal{E}[\mathbf{u}] = \sum |u_n|^2$. Indeed, using Eq.~(\ref{eqDTL_b}) we obtain
	\begin{equation}
	\label{eqA_DTL_b}
	\begin{array}{rcl}
	\mathcal{E}[\mathbf{u}']-\mathcal{E}[\mathbf{u}]
	& = & |u'_{n}|^2+|u'_{n+1}|^2-|u_{n}|^2-|u_{n+1}|^2 
	\\[3pt]
	& = & \mathrm{Re}\big( (u_{n}^{\prime *}+u_{n}^*)(u'_{n}-u_{n})
	+(u_{n+1}^{\prime *}+u_{n+1}^*)(u'_{n+1}-u_{n+1}) \big)
	\\[3pt]
	\displaystyle
	& = & 
	2 \Delta \tau_n \mathrm{Re} \big( v_n^*B_n^+[\mathbf{v}] 
	+v_{n+1}^* B_{n+1}^-[\mathbf{v}] \big) = 0,
	\end{array}
	\end{equation}
where we denoted $\mathbf{v} = (\mathbf{u}+\mathbf{u}')/2$ and the last equality follows after substituting the explicit form (\ref{eq1B_2}) of the quadratic terms; see also Eq.~(\ref{eq1B_3_Pi}).

Substituting Eq.~(\ref{eq1B_2}) into implicit relations (\ref{eqDTL_b}) and recalling that $u'_m = u_m$ for $m \ne n,n+1$, we write
	\begin{equation}
	\label{eqA_DTL_b2}
	\begin{array}{rcl}
	u'_{n} & = &  \displaystyle
	u_{n}+\frac{i\Delta \tau_n}{4} \Big( 2 k_{n+1}u_{n+2}\big(u_{n+1}+u'_{n+1}\big)^*
	+k_{n}\big(u_{n+1}+u'_{n+1}\big)u_{n-1}^* \Big), \\[9pt]
	u'_{n+1} & = & \displaystyle
	u_{n+1}+\frac{i\Delta \tau_n}{4} 
	\Big( -2k_{n+1}u_{n+2}\big(u_{n}+u'_{n}\big)^*+k_{n}\big(u_{n}+u'_{n}\big)u_{n-1}\Big).
	\end{array}
	\end{equation}
Substituting the second expression into the right-hand side of the first expression, yields a linear equation with respect to $u'_n$ and its complex conjugate. It is solved analytically by separating the real and imaginary parts. Similarly, one finds $u'_{n+1}$ after substituting the first expression of Eq.~(\ref{eqA_DTL_b2}) into the right-hand side of the second expression.

\subsection{Computer algorithm for the discrete-time model}
\label{app_alg}

Here we repeat the full algorithm described in Section~\ref{subsec_DTSM}, reformulated in a way that is convenient for computer implementation. In the initialization step, one is given the cutoff shell $N$, the boundary state (\ref{eq1_7}), the initial condition (\ref{eq1_8}) and the terminal time $0 < T < \infty$. Define the shell times $\tau_n = 0$ for $n = 1,\ldots,N$, with the respective variables $u_n = a_n$.
The algorithm then proceeds by iterating the following steps:
\begin{itemize}
\item[(a)] Choose $n$ as the smallest shell index for which $\tau_n$ attains the minimum among $\tau_1,\ldots,\tau_N$.
\item[(b)] If $n = 1$, then assign a new value to the variable $u_1$ according to Eq.~(\ref{eq1_10BC}).
\item[(c)] Choose time step $\Delta\tau_n$ from Eqs.~(\ref{eqDTS_TSt}) and (\ref{eqDTS_0T}) and update the time $\tau_n \mapsto \tau_n+\Delta\tau_n$.
\item[(d)] Compute the variables $u'_n$ and $u'_{n+1}$ according to Eq.~(\ref{eqDTL_b}); see also Appendix~\ref{subsec_A1}. When $n = N$, the value $u'_{n+1}$ is ignored; see Eq.~(\ref{eq1_7cut}). 
\item[(e)] Assign a new value to the variable $u_n$ according to Eq.~(\ref{eq1_10}). Also, set $u_{n+1} = u'_{n+1}$.
\item[(f)] The algorithm stops if $\tau_n = T$ for all $n = 1,\ldots,N$. Otherwise, repeat the algorithm starting from step (a).
\end{itemize}

Note that the time steps are random, since they are expressed in terms of the turnover time~(\ref{eqDTS_0T}), 
which depends on random state variables. This turnover time is bounded from below as
\begin{equation}
T_n \ge \frac{1}{k_n}\big(\mathcal{E} + |b_0|^2 + |b_{-1}|^2\big)^{-1/2},
\label{Aeq_D0}
\end{equation}
where $\mathcal{E} = \sum |u_n|^2$ is the energy.
Because the time steps~(\ref{eqDTS_TSt}) are defined as an $\varepsilon$-fraction of $T_n$, it follows that the algorithm converges in a finite number of steps, provided that the energy does not blow up. 
From a physical point of view, it is natural to assume that a regularization is meaningful only if it does not lead to blowup; in this case, the algorithm converges in a finite---though random---number of steps.

Owing to the conservative form of the nonlinearity, the absence of blowup for the truncated fluctuating Sabra model~(\ref{eq1_5}) can be established by using the energy as a Lyapunov function; see, e.g.,~\cite{mao2007stochastic}. By the same reasoning, this global regularity is expected to extend to the corresponding discrete-time algorithm.
On the other hand, the nonlinear form of the noise term in the LES shell model~(\ref{eq1_5SGS}) may lead to blowup, and we observed numerical indications of such behavior when the noise coefficient $C$ is taken large compared to the dissipation parameter $D$. However, when $C$ is not large, no blowup behavior is observed numerically.

\subsection{Anomalous scaling in discrete-time Sabra model}
\label{subsecA_2}

Developed turbulence is traditionally characterized by the existence of an inertial interval of scales, which separates large scales corresponding to external forcing from small dissipative scales. This inertial interval is characterized by a power-law scaling of the velocity moments, which are also called structure functions~\cite{frisch1999turbulence}. 
Let us consider the discrete-time Sabra model with $\varepsilon = 1$, $\mathrm{Re} = 10^7$ and the remaining parameters set as in Section~\ref{subsec_NSpSt}. In our analysis we define instantaneous structure functions as 
	\begin{equation}
	\label{eqA2_1}
	S_p(\ell_n) = \langle |u_n(T)|^p \rangle, 
	\end{equation}
where the velocity is taken at the final time $T = 1.5$ and the averaging is performed with respect to realizations of noise. 

Structure functions computed numerically using $4 \times 10^4$ simulations are presented in Fig.~\ref{figA1}(a), where we plotted the logarithms $\log_2 S_p$ for $p = 1,\ldots,6$. These graphs confirm the existence of the inertial interval with the power-law scaling $S_p(\ell_n) \propto \ell_n^{\zeta_p} = 2^{-n\zeta_p}$, as shown by the light-red straight lines. The measured exponents $\zeta_p$ are shown by red circles in Fig.~\ref{figA1}(b). Their nonlinear dependence on the moment order $p$ (opposed to Kolmogorov's K41 prediction $\zeta_p = p/3$ shown by the straight blue line) is the well-known characteristic of intermittent behavior. For comparison, we also plotted the exponents $\zeta_p$ corresponding to the original Sabra model as reported in \cite{de2024extreme}.

\begin{figure}[tp]
\centering
\includegraphics[width=0.88\textwidth]{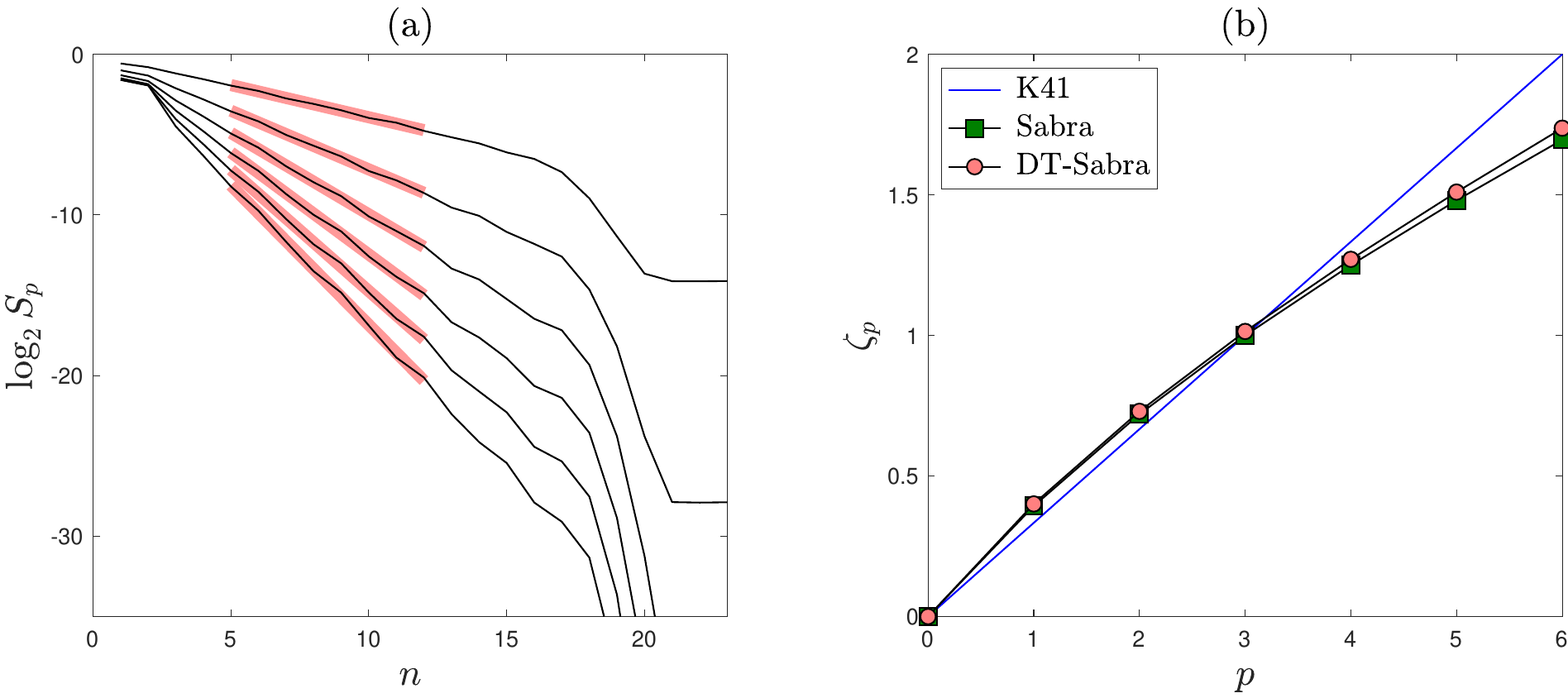}
\caption{(a) Logarithms $\log_2 S_p$ of structure functions in the discrete-time model. (b) Scaling exponents.}
\label{figA1}
\end{figure}

Our analysis not only confirms that our discrete-time shell model correctly reproduces the intermittency of developed turbulence. Note that our definition of structure functions (\ref{eqA2_1}) differs from the traditional one, where averaging is performed over long time. 
Thus, in the context of spontaneous stochasticity, the anomalous scaling laws are instantaneous (rather than long-term) characteristics of turbulent solutions; see also \cite{mailybaev2016spontaneously,bandak2024spontaneous}.

\subsection{Composition and sum of random maps}
\label{subsecA_3}

Here we recall some basic relations for random maps; see \cite{taylor2001user} for more details. Let $\varphi: \mathbf{s} \mapsto \mathbf{s}'$ be a random map, whose transition probability is determined by the probability kernel $\Phi(\mathbf{s}'|\mathbf{s})$. This kernel represents a probability measure for possible outcomes $\mathbf{s}'$ given an initial state $\mathbf{s}$. In particular, when the map $\varphi$ is deterministic, the respective kernel is given by the Dirac measure: $\Phi(\mathbf{s}'|\mathbf{s}) = \delta\big(\mathbf{s}'-\varphi(\mathbf{s})\big)$. Let $\varphi$ and $\psi$ be two statistically independent random maps with the respective probability kernels $\Phi$ and $\Psi$. Then the probability kernel for the composition $\psi \circ \varphi$ is denoted by $\Psi \circ \Phi$ and expressed as 
	\begin{equation}
	\label{eqA3_1}
	\Psi \circ \Phi (\mathbf{s}'|\mathbf{s}) = \int \Psi(\mathbf{s}'|\mathbf{s}'') \Phi (\mathbf{s}''|\mathbf{s}) \, d\mathbf{s}''.
	\end{equation}
Similarly, the probability kernel for the sum $\psi+\varphi$ is denoted by $\Psi * \Phi$ and expressed as 
	\begin{equation}
	\label{eqA3_2}
	\Psi * \Phi (\mathbf{s}'|\mathbf{s}) = \int \Psi(\mathbf{s}''|\mathbf{s}) \Phi (\mathbf{s}'-\mathbf{s}''|\mathbf{s}) \, d\mathbf{s}''.
	\end{equation}
This relation represents the convolution of probability kernels.

\vspace{2mm}\noindent\textbf{Acknowledgments.} 
This work was supported by CNPq - Conselho Nacional de Desenvolvimento Cient\'ifico e Tecnol\'ogico grant 308721/2021-7, FAPERJ - Fundac\~ao Carlos Chagas Filho de Amparo \`a Pesquisa do Estado do Rio de Janeiro grant E-26/201.054/2022, and CAPES MATH-AmSud project CHA2MAN. 

\vspace{2mm}\noindent\textbf{Data availability. } 
The data that support the findings of this article are openly available \cite{data}.

\bibliographystyle{plain}
\bibliography{refs}

@book{mao2007stochastic,
  title={Stochastic differential equations and applications},
  author={Mao, X.},
  year={2007},
  publisher={Elsevier}
}

@article{liao2025noise,
  title={Noise-expansion cascade: an origin of randomness of turbulence},
  author={Liao, S. and Qin, S.},
  journal={Journal of Fluid Mechanics},
  volume={1009},
  pages={A2},
  year={2025},
  publisher={Cambridge University Press}
}

@misc{data,
  author={The data of the figures are available under the link: \url{https://doi.org/10.5281/zenodo.17250259}.},
}

@article{eyink2020renormalization,
  title={Renormalization group approach to spontaneous stochasticity},
  author={Eyink, G. L. and Bandak, D.},
  journal={Physical Review Research},
  volume={2},
  number={4},
  pages={043161},
  year={2020},
  publisher={APS}
}

@article{de2024extreme,
  title={Extreme statistics and extreme events in dynamical models of turbulence},
  author={de Wit, X. M. and Ortali, G. and Corbetta, A. and Mailybaev, A. A. and Biferale, L. and Toschi, F.},
  journal={Physical Review E},
  volume={109},
  number={5},
  pages={055106},
  year={2024},
  publisher={APS}
}

@article{drivas2024statistical,
  title={{Statistical determinism in non-Lipschitz dynamical systems}},
  author={Drivas, T. D. and Mailybaev, A. A. and Raibekas, A.},
  journal={Ergodic Theory and Dynamical Systems},
  volume={44},
  number={7},
  pages={1856--1884},
  year={2024},
  publisher={Cambridge University Press}
}

@article{barlet2025spontaneous,
  title={{Spontaneous stochasticity in a 3d Weierstrass-ABC flow}},
  author={Barlet, A. and Cheminet, A. and Dubrulle, B. and Mailybaev, A. A.},
  journal={Nonlinearity},
  volume={39},
  number={1},
  pages={015018},
  year={2026}
}

@article{ortiz2025spontaneous,
  title={{Spontaneous stochasticity in the fluctuating Navier-Stokes equations on a logarithmic lattice}},
  author={Ortiz, E. and Campolina, C. S and Mailybaev, A. A.},
  journal={Preprint arXiv:2507.03196},
  year={2025}
}

@article{mailybaev2025rg,
  title={{RG analysis of spontaneous stochasticity on a fractal lattice: stability and bifurcations}},
  author={Mailybaev, A. A.},
  journal={Journal of Statistical Physics},
  volume={192},
  number={3},
  pages={1--22},
  year={2025},
  publisher={Springer}
}

@article{mailybaev2015stochastic,
  title={{Stochastic anomaly and large Reynolds number limit in hydrodynamic turbulence models}},
  author={Mailybaev, A. A.},
  journal={Preprint arXiv:1508.03869},
  year={2015}
}

@article{fjordholm2016computation,
  title={On the computation of measure-valued solutions},
  author={Fjordholm, U. S. and Mishra, S. and Tadmor, E.},
  journal={Acta Numerica},
  volume={25},
  pages={567--679},
  year={2016},
  publisher={Cambridge University Press}
}

@article{de2010admissibility,
  title={{On admissibility criteria for weak solutions of the Euler equations}},
  author={De Lellis, C. and Szekelyhidi, L.},
  journal={Archive for Rational Mechanics and Analysis},
  volume={195},
  pages={225--260},
  year={2010},
  publisher={Springer}
}

@article{daneri2021non,
  title={{Non-uniqueness for the Euler equations up to Onsager’s critical exponent}},
  author={Daneri, S. and Runa, E. and Szekelyhidi, L.},
  journal={Annals of PDE},
  volume={7},
  number={1},
  pages={8},
  year={2021},
  publisher={Springer}
}

@article{boffetta2017chaos,
  title={Chaos and predictability of homogeneous-isotropic turbulence},
  author={Boffetta, G. and Musacchio, S.},
  journal={Physical Review Letters},
  volume={119},
  number={5},
  pages={054102},
  year={2017},
  publisher={APS}
}

@article{mailybaev2023spontaneously,
  title={{Spontaneously stochastic Arnold’s cat}},
  author={Mailybaev, A. A. and Raibekas, A.},
  journal={Arnold Mathematical Journal},
  volume={9},
  number={3},
  pages={339--357},
  year={2023},
  publisher={Springer}
}

@article{bernard1998slow,
  title={Slow modes in passive advection},
  author={Bernard, D. and Gawedzki, K. and Kupiainen, A.},
  journal={Journal of Statistical Physics},
  volume={90},
  number={3},
  pages={519--569},
  year={1998},
  publisher={Springer}
}

@article{kupiainen2003nondeterministic,
  title={Nondeterministic dynamics and turbulent transport},
  author={Kupiainen, A.},
  journal={Annales Henri Poincar{\'e}},
  volume={4},
  number={2},
  pages={713--726},
  year={2003}
}

@article{falkovich2001particles,
  title={Particles and fields in fluid turbulence},
  author={Falkovich, G. and Gawedzki, K. and Vergassola, M.},
  journal={Reviews of Modern Physics},
  volume={73},
  number={4},
  pages={913},
  year={2001},
  publisher={APS}
}

@article{eyink2024space,
  title={{Space-time statistical solutions of the incompressible Euler equations and Landau-Lifshitz fluctuating hydrodynamics}},
  author={Eyink, G. L. and Peng, L.},
  journal={Nonlinearity},
  volume={38},
  number={8},
  pages={085011},
  year={2025},
}

@article{mailybaev2017toward,
  title={{Toward analytic theory of the Rayleigh--Taylor instability: lessons from a toy model}},
  author={Mailybaev, A. A.},
  journal={Nonlinearity},
  volume={30},
  number={6},
  pages={2466--2484},
  year={2017},
  publisher={IOP Publishing}
}

@article{biferale2018rayleigh,
  title={{Rayleigh-Taylor turbulence with singular nonuniform initial conditions}},
  author={Biferale, L. and Boffetta, G. and Mailybaev, A. A. and Scagliarini, A.},
  journal={Physical Review Fluids},
  volume={3},
  number={9},
  pages={092601(R)},
  year={2018},
  publisher={APS}
}

@article{eyink1996turbulence,
  title={Turbulence noise},
  author={Eyink, G. L.},
  journal={Journal of Statistical Physics},
  volume={83},
  number={5},
  pages={955--1019},
  year={1996},
  publisher={Springer}
}

@article{leith1972predictability,
  title={Predictability of turbulent flows},
  author={Leith, C. E. and Kraichnan, R. H.},
  journal={Journal of Atmospheric Sciences},
  volume={29},
  number={6},
  pages={1041--1058},
  year={1972}
}

@article{lorenz1969predictability,
  title={The predictability of a flow which possesses many scales of motion},
  author={Lorenz, E. N.},
  journal={Tellus},
  volume={21},
  number={3},
  pages={289--307},
  year={1969},
  publisher={Taylor \& Francis}
}

@article{thalabard2020butterfly,
  title={From the butterfly effect to spontaneous stochasticity in singular shear flows},
  author={Thalabard, S. and Bec, J. and Mailybaev, A. A.},
  journal={Communications Physics},
  volume={3},
  number={1},
  pages={122},
  year={2020},
  publisher={Nature Publishing Group UK London}
}

@book{taylor2001user,
  title={A user's guide to measure-theoretic probability},
  author={Pollard, D.},
  year={2002},
  publisher={Cambridge University Press},
  address = {Cambridge}
}

@book{arnold1992ordinary,
  title={Ordinary differential equations},
  author={Arnold, V. I.},
  year={1992},
  publisher={Springer},
  address={New York}
}

@article{mailybaev2023spontaneous,
  title={Spontaneous stochasticity and renormalization group in discrete multi-scale dynamics},
  author={Mailybaev, A. A. and Raibekas, A.},
  journal={Communications in Mathematical Physics},
  volume={401},
  number={3},
  pages={2643--2671},
  year={2023},
  publisher={Springer}
}

@article{mailybaev2024rg,
  title={{RG approach to the inviscid limit for shell models of turbulence}},
  author={Mailybaev, A. A.},
  journal={Nonlinearity},
  volume={38},
  pages={085010},
  year={2025}
}

@book{oksendal2013stochastic,
  title={Stochastic differential equations: an introduction with applications},
  author={Oksendal, B.},
  year={2013},
  publisher={Springer}
}

@article{garratt1994atmospheric,
  title={The atmospheric boundary layer},
  author={Garratt, J. R.},
  journal={Earth-Science Reviews},
  volume={37},
  number={1-2},
  pages={89--134},
  year={1994},
  publisher={Elsevier}
}

@book{landau1959fluid,
  title={Fluid Mechanics},
  author={Landau, L. D. and Lifshitz, E. M.},
  year={1959},
  publisher={Pergamon Press, London}
}

@article{dombre1998intermittency,
  title={{Intermittency, chaos and singular fluctuations in the mixed Obukhov-Novikov shell model of turbulence}},
  author={Dombre, T. and Gilson, J.-L.},
  journal={Physica D: Nonlinear Phenomena},
  volume={111},
  number={1-4},
  pages={265--287},
  year={1998},
  publisher={Elsevier}
}

@article{mailybaev2012renormalization,
  title={Renormalization and universality of blowup in hydrodynamic flows},
  author={Mailybaev, A. A.},
  journal={Physical Review E},
  volume={85},
  number={6},
  pages={066317},
  year={2012},
  publisher={APS}
}

@article{bandak2022dissipation,
  title={Dissipation-range fluid turbulence and thermal noise},
  author={Bandak, D. and Goldenfeld, N. and Mailybaev, A. A. and Eyink, G.},
  journal={Physical Review E},
  volume={105},
  number={6},
  pages={065113},
  year={2022},
  publisher={APS}
}

@article{bandak2024spontaneous,
  title={Spontaneous stochasticity amplifies even thermal noise to the largest scales of turbulence in a few eddy turnover times},
  author={Bandak, D. and Mailybaev, A. A. and Eyink, G. L. and Goldenfeld, N.},
  journal={Physical Review Letters},
  volume={132},
  number={10},
  pages={104002},
  year={2024},
  publisher={APS}
}

@article{mailybaev2016spontaneous,
  title={Spontaneous stochasticity of velocity in turbulence models},
  author={Mailybaev, A. A.},
  journal={Multiscale Modeling \& Simulation},
  volume={14},
  number={1},
  pages={96--112},
  year={2016},
  publisher={SIAM}
}

@book{pope2000turbulent,
  title={Turbulent flows},
  author={Pope, S. B.},
  year={2000},
  publisher={Cambridge University Press, New York}
}

@article{mailybaev2016spontaneously,
  title={Spontaneously stochastic solutions in one-dimensional inviscid systems},
  author={Mailybaev, A. A.},
  journal={Nonlinearity},
  volume={29},
  number={8},
  pages={2238--2252},
  year={2016},
  publisher={IOP Publishing}
}

@article{ruelle1979microscopic,
  title={Microscopic fluctuations and turbulence},
  author={Ruelle, D.},
  journal={Physics Letters A},
  volume={72},
  number={2},
  pages={81--82},
  year={1979},
  publisher={Elsevier}
}

@article{l1998improved,
  title={Improved shell model of turbulence},
  author={L'vov, V. S. and Podivilov, E. and Pomyalov, A. and Procaccia, I. and Vandembroucq, D.},
  journal={Phys. Rev. E},
  volume={58},
  number={2},
  pages={1811},
  year={1998},
  publisher={APS}
}

@article{biferale2003shell,
  title={Shell models of energy cascade in turbulence},
  author={Biferale, L.},
  journal={Annual Review of Fluid Mechanics},
  volume={35},
  pages={441--468},
  year={2003},
}

@book{frisch1999turbulence,
  title={{Turbulence: the Legacy of A.N. Kolmogorov}},
  author={Frisch, U.},
  year={1995},
  publisher={Cambridge University Press},
  address = {Cambridge}
}

\end{document}